\documentclass[a4paper,USenglish,cleveref,pdfa,authorcolumns]{lipics-v2021}

\pdfoutput=1
\nolinenumbers

\usepackage{preamble}

\title{Broadcasting Under Structural Restrictions\footnote{An extended abstract of this work was presented at the
50th International Symposium on Mathematical Foundations of Computer Science (MFCS 2025)~\cite{mfcs/EgamiGHKLMNOV025}.}}

\author{Yudai Egami}
{Kyushu University, Fukuoka, Japan}
{egami.yudai.637@s.kyushu-u.ac.jp}
{}
{}

\author{Tatsuya Gima}
{Hokkaido University, Sapporo, Japan}
{gima@ist.hokudai.ac.jp}
{https://orcid.org/0000-0003-2815-5699}
{JSPS KAKENHI Grant Numbers
JP24K23847, 
JP25K03077. 
}

\author{Tesshu Hanaka}
{Kyushu University, Fukuoka, Japan
\and \url{https://sites.google.com/view/tesshu-hanaka/home}}
{hanaka@inf.kyushu-u.ac.jp}
{https://orcid.org/0000-0001-6943-856X}
{JSPS KAKENHI Grant Numbers
JP21H05852, 
JP21K17707, 
JP22H00513, 
JP23H04388, 
JP25K03077, 
and JST, CRONOS, Japan Grant Number JPMJCS24K2.
}

\author{Yasuaki Kobayashi}
{Hokkaido University, Sapporo, Japan}
{koba@ist.hokudai.ac.jp}
{https://orcid.org/0000-0003-3244-6915}
{JSPS KAKENHI Grant Numbers
JP20H00595,
JP23K28034,
JP24H00686,
and JP24H00697.}

\author{Michael Lampis}
{Universit\'{e} Paris-Dauphine, PSL University, CNRS UMR7243, LAMSADE, Paris, France
\and \url{https://www.lamsade.dauphine.fr/~mlampis/index.html}}
{michail.lampis@dauphine.fr}
{https://orcid.org/0000-0002-5791-0887}{}

\author{Valia Mitsou}
{Universit\'{e} Paris Cit\'{e}, IRIF, CNRS, 75205, Paris, France}
{vmitsou@irif.fr}{}{}

\author{Edouard Nemery}
{Universit\'{e} Paris-Dauphine, PSL University, CNRS UMR7243, LAMSADE, Paris, France}
{edouard.nemery@dauphine.psl.eu}
{https://orcid.org/0009-0007-6977-9330}{}

\author{Yota Otachi}
{Institute of Science Tokyo, Tokyo, Japan}
{otachi@comp.isct.ac.jp}
{https://orcid.org/0000-0002-0087-853X}
{JSPS KAKENHI Grant Numbers
JP21K11752, 
JP22H00513, 
JP24H00697, 
JP25K03076, 
JP25K03077. 
}

\author{Manolis Vasilakis}
{University of Southern Denmark, Odense, Denmark \and \url{https://manolisvasilakis.github.io/}}
{vasilakis@imada.sdu.dk}
{https://orcid.org/0000-0001-6505-2977}{}

\author{Daniel Vaz}
{LIGM, Université Gustave Eiffel, CNRS, ESIEE Paris, 77454 Marne-la-Vallée, France
\and \url{https://www.dvaz.net/}}
{daniel.ramosvaz@esiee.fr}
{https://orcid.org/0000-0003-2224-2185}{}

\authorrunning{Y. Egami et al.}

\Copyright{Yudai Egami, Tatsuya Gima, Tesshu Hanaka, Yasuaki Kobayashi, Michael Lampis, Valia Mitsou, Edouard Nemery, Yota Otachi, Manolis Vasilakis, and Daniel Vaz}

\ccsdesc[500]{Theory of computation~Parameterized complexity and exact algorithms}

\keywords{Parameterized Complexity, Structural Graph Parameters, Telephone Broadcast}

\category{} 

\funding{This work is partially supported by ANR project ANR-21-CE48-0022 (S-EX-AP-PE-AL).}

\EventEditors{Pawe\l{} Gawrychowski, Filip Mazowiecki, and Micha\l{} Skrzypczak}
\EventNoEds{3}
\EventLongTitle{50th International Symposium on Mathematical Foundations of Computer Science (MFCS 2025)}
\EventShortTitle{MFCS 2025}
\EventAcronym{MFCS}
\EventYear{2025}
\EventDate{August 25--29, 2025}
\EventLocation{Warsaw, Poland}
\EventLogo{}
\SeriesVolume{345}
\ArticleNo{44}

\begin{document}

\maketitle
\begin{abstract}
In the \textsc{Telephone Broadcast} problem we are given a graph $G=(V,E)$ with a designated source
vertex $s \in V$. Our goal is to transmit a message, which is initially known
only to $s$, to all vertices of the graph by using a process where in each
round an informed vertex may transmit the message to one of its uninformed
neighbors. The optimization objective is to minimize the number of rounds. 

Following up on several recent works, we investigate the structurally
parameterized complexity of \textsc{Telephone Broadcast}. In particular, we first strengthen existing
NP-hardness results by showing that the problem remains NP-complete on graphs
of bounded tree-depth and also on cactus graphs which are one vertex deletion
away from being path forests. Motivated by this (severe) hardness, we study
several other parameterizations of the problem and obtain FPT algorithms
parameterized by vertex integrity (generalizing a recent FPT algorithm
parameterized by vertex cover by Fomin, Fraigniaud, and Golovach~[TCS 2024])
and by distance to clique, as well as FPT approximation algorithms
parameterized by clique-cover and cluster vertex deletion. Furthermore, we
obtain structural results that relate the length of the optimal broadcast
protocol of a graph $G$ with its pathwidth and tree-depth. By presenting a
substantial improvement over the best previously known bound for pathwidth
(Aminian, Kamali, Seyed-Javadi, and Sumedha~[ICALP 2025]) we exponentially
improve the approximation ratio achievable in polynomial time on graphs of
bounded pathwidth from $\mathcal{O}(4^\mathrm{pw})$ to $\mathcal{O}(\mathrm{pw})$.

\end{abstract}


\section{Introduction}\label{sec:introduction}

Broadcasting in graphs is a well-studied field, with one of the most basic models being the \emph{telephone model}~\cite{networks/HedetniemiHL88}:
at the start of the broadcasting process exactly one vertex of the graph (the \emph{source} vertex) is informed (i.e., knows the message),
and at the beginning of each round each informed vertex can send the message to at most one of its uninformed neighbors
which receives the message by the end of that round.
To fix some notation, let $G$ be a connected graph on $n$ vertices and $s \in V(G)$ denote the source vertex.
Then, $b(G,s)$ denotes the minimum number of rounds needed to disseminate the message to all vertices of $G$ under the previously described model.
As an example, if $G$ is a path and $s$ is one of the degree-$1$ vertices,
then $b(G,s) = n-1$,
while if $G$ is a complete graph, then $b(G,s) = \lceil \log n \rceil$.
Notice that for any graph $G$ it holds that $\lceil \log n \rceil \le b(G,s) \le n-1$,
and due to the previous examples there exist graphs for which those bounds are tight.

In this paper we study several structural parameterizations%
\footnote{Throughout the paper we assume that the reader is familiar with the basics of parameterized complexity,
as given in standard textbooks~\cite{books/CyganFKLMPPS15}.}
of the {\TB} problem which asks, given a graph $G$, a source vertex $s \in V(G)$, and an integer $t$, whether $b(G,s) \le t$.
In 1981, Slater, Cockayne, and Hedetniemi~\cite{siamcomp/SlaterCH81} introduced {\TB}
and showed that it is NP-complete in general but polynomial-time solvable on trees.
Among several results obtained in the intensive study of the problem, the most relevant to us are the very recent ones due to
Fomin et al.~\cite{tcs/FominFG24}, Tale~\cite{tcs/Tale25}, and Aminian et al.~\cite{icalp/AminianKJS25}, described below.

Fomin, Fraigniaud, and Golovach~\cite{tcs/FominFG24} initiated the study of the parameterized complexity of \TB.
They developed a $3^n n^{\bO(1)}$-time algorithm, where $n$ is the number of vertices,
as well as showed that the problem is fixed-parameter tractable parameterized
by feedback edge set number, by vertex cover number, and by the number of vertices minus the number of rounds;
notice that their first result, along with the fact that $n \le 2^t$,
further implies that the problem is FPT by the number of rounds and leads to an algorithm of double-exponential parametric dependence.
Among other open questions they posed, they asked the complexity of {\TB} parameterized by other structural parameters
and explicitly mentioned treewidth and feedback vertex set number.

Answering the previous question, Tale~\cite{tcs/Tale25} showed that {\TB} is NP-complete on graphs of feedback vertex set number~$1$
(and thus treewidth~$2$).
As a matter of fact, the proof shows NP-completeness on graphs of vertex deletion distance~$2$ to path forest (and thus pathwidth~$3$).
Regarding the double-exponential dependence parameterized by the number of rounds $t$,
as discussed in~\cite{tcs/Tale25}, a result of Papadimitriou and Yannakakis~\cite{jacm/PapadimitriouY82}
implies that {\TB} parameterized by $t$ does not admit an algorithm with parameter dependence better than double-exponential (under the ETH). Therefore, the bound obtained by the observation $n\le 2^t$ applied to the $3^nn^{\bO(1)}$-time algorithm of~\cite{tcs/FominFG24} is essentially optimal.



In a very recent work,
Aminian, Kamali, Seyed-Javadi, and Sumedha~\cite{icalp/AminianKJS25} studied {\TB} on cactus graphs and on graphs of bounded pathwidth.
They showed that the problem is NP-complete on cactus graphs of pathwidth~$2$,
while they presented a factor-$2$ approximation algorithm for cactus graphs.
Furthermore, they established a connection between the pathwidth $\pw$ of the input graph $G$ and $b(G,s)$,
by proving that for all $s \in V(G)$ it holds that $b(G,s) = \Omega (n^{4^{-(\pw+1)}})$.
A consequence of this result is that the polynomial-time approximation algorithm of Elkin and Kortsarz~\cite{jcss/ElkinK06}
achieves an approximation ratio of $\bO(4^{\pw})$ for graphs of pathwidth~$\pw$,
improving over the factor $\bO (\frac{\log n}{\log{\log n}})$ achieved in general graphs.

\begin{figure}[ht]
\centering
\scalebox{0.8}{ 
\begin{tikzpicture}[
  every node/.style={draw, rectangle, rounded corners, align=center, minimum width=2.8cm, font=\small},
  paraNP/.style={text=red!80!black,draw, rectangle,  double, sharp corners},
  XP/.style={text=blue, draw, rectangle, sharp corners},
  FPTAS/.style={text=orange,dash pattern=on 3pt off 2pt},
  FPT/.style={text=green!30!black,thick},
  Open/.style={text=black,dash pattern=on 3pt off 2pt,sharp corners,thick},
]

\node[Open] (ccn) at (-6.5,-0.5) {clique cover$^*$ \\ ($1+\varepsilon$)-approx.};
\node[Open] (cvd) at (-4,1) {cluster vertex deletion$^*$ \\ ($2+\varepsilon$)-approx.};
\node[FPT] (dtc) at (-6.5,-3.5) {distance to clique$^*$};


\node[paraNP] (cw) at (-0.5,4) {cliquewidth};
\node[paraNP] (tw) at (-0.5,2.5) {treewidth};
\node[paraNP] (pw) at (-0.5,1) {pathwidth};
\node[paraNP] (td) at (-0.5,-0.5) {tree-depth$^*$};
\node[FPT] (vi) at (-0.5,-2) {vertex integrity$^*$};
\node[FPT] (vc) at (-0.5,-3.5) {vertex cover~\cite{tcs/FominFG24}};

\node[paraNP] (fvs) at (3.5,1) {feedback vertex set};
\node[paraNP] (dtp) at (3.5,-0.5) {distance to path forest$^*$ \cite{tcs/Tale25}};
\node[FPT] (fes) at (7,-2) {feedback edge set~\cite{tcs/FominFG24}};


\draw [thick] (cw) -- (tw);
\draw [thick] (tw) -- (pw);
\draw [thick] (pw) -- (td);
\draw [thick] (tw) -- (fvs);
\draw [thick] (td) -- (vi);
\draw [thick] (vi) -- (vc);
\draw [thick] (ccn) -- (dtc);

\draw [thick] (cw) to[bend right=25] (cvd);
\draw [thick] (cvd) to[bend right=18] (vc);
\draw [thick] (cvd) to[bend left=18] (dtc);
\draw [thick] (fvs) -- (dtp);
\draw [thick] (pw) -- (dtp);
\draw [thick] (fvs) to[bend left=25] (fes);
\draw [thick] (dtp) -- (vc);

\node[paraNP] at (7,3.2) {\textcolor{red!80!black}{paraNP-complete}};
\node[FPT] at (7,2.5) {\textcolor{green!30!black}{FPT}};
\node[Open] at (7,1.8) {FPT approximable};

\end{tikzpicture}
}
\caption{Parameterized complexity of {\TB} with respect to structural graph parameters.
The connection between two parameters represents that the upper parameter $p$ is bounded by some
computable function $f(\cdot)$ of the lower parameter $q$, i.e., $p \le f (q)$.
The parameters marked with $*$ are studied in this paper.
The double, rounded, and dotted rectangles indicate paraNP-complete, fixed-parameter tractable, and parameterized approximable cases, respectively.
The exact parameterized complexity of the parameterizations in the last case remains unsettled.
}

\label{fig:graph-parameters}
\end{figure}

\subparagraph{Our contribution.}
Our aim in this paper is to thoroughly study {\TB} under the perspective of parameterized complexity;
for an overview of the relationship between the mentioned parameters we refer to \cref{fig:graph-parameters},
while \cref{tab:results-overview} summarizes our main results.

\begin{table}[tbh]
    \centering
    \footnotesize
    \renewcommand{\arraystretch}{1.1}
    \begin{tabularx}{\textwidth}{@{}>{\raggedright\arraybackslash}X
        >{\raggedright\arraybackslash}X l@{}}
        \toprule
        Parameter or graph class & Result & Reference \\
        \midrule
        Cactus graphs, distance~$1$ to path forest
            & NP-completeness
            & \cref{thm:hardness:cactus} \\
        Tree-depth at most~$6$
            & NP-completeness
            & \cref{thm:hardness:treedepth} \\
        Vertex integrity $k$
            & Exact FPT in $2^{\bO(k^5)}n^{\bO(1)}$ time
            & \cref{thm:fpt:vi} \\
        Distance to clique $k$
            & Exact FPT in $2^{\bO(k^2)}n^{\bO(1)}$ time
            & \cref{thm:fpt:dtc} \\
        Treewidth $\tw$ and number of rounds $t$
            & Exact FPT in $2^{\bO(t\cdot\tw)}n^{\bO(1)}$ time
            & \cref{thm:fpt:twt} \\
        Clique cover number
            & FPT $(1+\varepsilon)$-approximation scheme
            & \cref{thm:FPTAS:mmc} \\
        Cluster vertex deletion number
            & FPT $(2+\varepsilon)$-approximation scheme
            & \cref{thm:FPTAS:cvd} \\
        Tree-depth $\td$
            & Lower bound: $b(G,s) \ge n^{1/\td}/\td$
            & \cref{thm:structural_lb:td} \\
        Pathwidth $\pw$
            & Lower bound: $b(G,s) = \Omega(n^{1/(2\pw+2)})$
            & \cref{corollary:structural_lb:pw} \\
        Pathwidth $\pw$
            & Polynomial-time $\bO(\pw)$-approximation
            & \cref{corollary:structural_lb:apx_ratio} \\
        \bottomrule
    \end{tabularx}
    \caption{Summary of our main results.}
    \label{tab:results-overview}
\end{table}

First we show in \cref{sec:np_hardness} that the problem remains NP-complete even for very restricted families of graphs,
namely (i) cactus graphs of vertex deletion distance~$1$ to path forest (\cref{thm:hardness:cactus}),
and (ii) graphs of constant tree-depth (\cref{thm:hardness:treedepth}).
The reduction employed in \cref{thm:hardness:cactus} follows along the lines of the one presented by Tale~\cite{tcs/Tale25},
starting from a variation of \textsc{$3$-Partition}.
However, by selecting a more appropriate variant of \textsc{$3$-Partition} we are able to significantly simplify both the reduction itself and the constructed graph,
thus improving over the initial result by Tale~\cite{tcs/Tale25} and the subsequent one by Aminian et al.~\cite{icalp/AminianKJS25} (in the sense that graphs with vertex deletion distance~$1$ from a path forest, as given in our reduction, are a strict subclass of the snowflake graphs given in \cite{icalp/AminianKJS25}).
More importantly, by modifying said reduction, we further prove that the problem remains NP-hard for graphs of \emph{constant tree-depth}.
Note that tree-depth is a parameter not covered by the previous reductions, which heavily relied on constructions of long paths (which inherently have large tree-depth).

Our negative results extend the intractability of the problem from bounded pathwidth to bounded tree-depth instances. It is therefore natural to ask which (more restrictive) parameterizations \emph{do} render the problem tractable.
We present several positive results in this direction in \cref{sec:fpt-alg}.
In particular, in \cref{thm:fpt:vi} we show that the problem is FPT parameterized by vertex integrity,
which nicely complements the paraNP-hardness in case of parameterization by tree-depth.
To obtain this result,
we generalize some of the ideas employed by Fomin et al.~\cite{tcs/FominFG24} to show fixed-parameter tractability by vertex cover number.
Next, in \cref{thm:fpt:dtc} we consider graphs which are $k$ vertex deletions away from being a clique
and we develop a $2^{\bO(k^2)} n^{\bO(1)}$-time algorithm.
Moving on, we consider the parameterization by both $t$ and the treewidth $\tw$ of the input graph,
and in \cref{thm:fpt:twt} we develop a single-exponential $2^{\bO(t \cdot \tw)} n^{\bO(1)}$-time algorithm based on dynamic programming (contrast this with the fact that the problem is intractable when parameterized solely by treewidth and requires a double-exponential dependence when parameterized solely by $t$).
As our last algorithmic results, in \cref{sec:para-approx} we develop FPT approximation algorithms when parameterized
by clique cover number and by cluster vertex deletion number, in \cref{thm:FPTAS:mmc,thm:FPTAS:cvd} respectively.
With these results we clearly trace the limits of tractability for sparse parameters (tree-depth, vertex integrity) but also improve our understanding of some dense graph parameters.

The $2^{\bO(t \cdot \tw)} n^{\bO(1)}$ algorithm is of particular note because
it identifies a small island of tractability for \TB, namely, it implies that
\TB\ is polynomial-time solvable for graphs of bounded treewidth (XP) whenever
we are interested in a \emph{fast} broadcast schedule ($t=\bO(\log n)$). This
is interesting, because in general the problem seems quite intractable for most
parameters and because all hardness reductions we mentioned construct instances
which require a large (polynomial in $n$) number of rounds. This further
motivates the question of which graph families may admit a schedule with
\emph{few} rounds, which, as we mention below, has been extensively studied and
also has algorithmic implications in terms of approximation.  Closing the
paper, in \cref{sec:structural_lb} we improve our understanding of this (structural)
question by more carefully studying the relationship between the minimum number of necessary
rounds and the pathwidth (\cref{thm:structural_lb:pw_plus_length}) and tree-depth (\cref{thm:structural_lb:td}) of a graph.
In particular,  we show as a main result (\cref{corollary:structural_lb:pw}) that $b(G,s) = \Omega (n^{(2\pw+2)^{-1}})$.
This bound exponentially improves over the corresponding bound of Aminian et
al.~\cite{icalp/AminianKJS25} (who showed that $b(G,s) = \Omega
(n^{4^{-(\pw+1)}})$). As an immediate application
(\cref{corollary:structural_lb:apx_ratio}), by applying an algorithm of Elkin
and Kortsarz~\cite{jcss/ElkinK06}, we obtain a polynomial-time algorithm with
approximation ratio $\bO(\pw)$ (compared to the ratio of $\bO(4^\pw)$ given in
\cite{icalp/AminianKJS25}).

\subparagraph{Related work.}
There has been an extensive line of research with regards to broadcasting under the telephone model in graphs.
The computational complexity of {\TB} has been studied for multiple classes of graphs,
with the problem being polynomial-time solvable for trees~\cite{siamcomp/SlaterCH81},
a subclass of cactus graphs~\cite{jco/CevnikZ17},
fully connected trees~\cite{join/GholamiHM23},
generalized windmill graphs (equivalently, graphs of twin-cover number~$1$)~\cite{algorithms/AmbashankarH15,jgaa/Damaschke24},
and unicyclic graphs~\cite{jco/HarutyunyanM08},
with the latter result having been generalized to the previously mentioned FPT
algorithm parameterized by feedback edge set number by Fomin et
al.~\cite{tcs/FominFG24}.  On the other hand, Jansen and
M\"{u}ller~\cite{tcs/JansenM95} showed that {\TB} remains NP-complete on
several restricted graph classes such as chordal and grid graphs.  We also
refer to some relevant surveys in~\cite{dam/FraigniaudL94,sp/Hromkovic05}.  In
another line of research, Grigni and Peleg~\cite{siamdm/GrigniP91} have
provided bounds on the minimum number of edges and the minimum maxdegree for
graphs $G$ with the property that $b(G,s) = \lceil \log n \rceil$ for all $s
\in V(G)$.

The problem is also well-studied in terms of its polynomial-time approximability.
Ravi~\cite{focs/Ravi94} presented a $\bO(\frac{\log^2n}{\log{\log n}})$-factor approximation algorithm,
while Kortsarz and Peleg~\cite{siamdm/KortsarzP95} presented (among other results) an algorithm that computes a broadcast protocol with
at most $2b(G,s) + \bO(\sqrt{n})$ rounds.
Subsequently, Bar-Noy, Guha, Naor, and Schieber~\cite{siamcomp/Bar-NoyGNS00} improved the approximation ratio to $\bO(\log n)$,
while Elkin and Kortsarz~\cite{siamcomp/ElkinK05} provided a \emph{combinatorial} approximation algorithm of the same approximation ratio,
as opposed to the algorithms of~\cite{siamcomp/Bar-NoyGNS00,focs/Ravi94} that involve solving LPs,
that additionally generalizes to directed graphs.
The current best approximation ratio of $\bO(\frac{\log n}{\log{\log n}})$ is due to Elkin and Kortsarz~\cite{jcss/ElkinK06};
as a matter of fact, the ratio of their algorithm is $\bO(\frac{\log n}{\log{b(G,s)}})$,
thus whenever $b(G,s) = \Omega (n^{\delta})$ for some constant $\delta > 0$,
they obtain a \emph{constant} approximation ratio $\bO(1 / \delta)$.
Better approximation ratios are known for specific graph classes~\cite{caldam/BhabakH15,join/BhabakH19,ciac/HarutyunyanH23,siamdm/KortsarzP95},
while there are also some inapproximability results~\cite{siamcomp/ElkinK05,approx/Schindelhauer00}.

There are several related problems in the literature (see e.g.~\cite{caldam/XuL24} or the survey~\cite{networks/HedetniemiHL88}),
and other models of communication have been considered, such as the \emph{open-line} and \emph{open-path} model~\cite{networks/Farley80,siamdm/KortsarzP95}.
An important generalization of {\TB} that has been studied is the \textsc{Telephone Multicast} problem,
where the objective is to minimize the number of rounds needed to inform a specific \emph{subset} of vertices of the input graph.
Most of the previously mentioned approximation algorithms (\cite{siamcomp/Bar-NoyGNS00,siamcomp/ElkinK05,jcss/ElkinK06,focs/Ravi94})
generalize to this problem resulting in improved approximation ratios,
while directed versions (and variations) of it have also been studied~\cite{algorithmica/ElkinK06,algorithmica/HathcockKR26}.

\section{Preliminaries}\label{sec:preliminaries}
Throughout the paper we use standard graph notations~\cite{Diestel17},
and we assume familiarity with the basic notions of parameterized complexity~\cite{books/CyganFKLMPPS15}.
For $x, y \in \Z$, let $[x, y] = \setdef{z \in \Z}{x \leq z \leq y}$, while $[x] = [1,x]$.

All graphs considered are undirected without loops.
Given a graph $G=(V,E)$ and a subset of its vertices $S \subseteq V$,
$G[S]$ denotes the subgraph induced by $S$ while $G - S$ denotes $G[V \setminus S]$.
We denote the set of connected components of $G$ by $\cc(G)$.
For two vertices $u,v \in V$, the distance between $u$ and $v$ in $G$, denoted by $d(u,v)$, is the length of a $u$-$v$ shortest path.
For $C \in \cc(G)$, $d_{C} (u,v)$ denotes the distance between $u$ and $v$ in the connected component $C$.
The diameter of a connected graph $G$, denoted by $\diam(G)$,
is the maximum among the lengths of its shortest paths, i.e., $\max_{u,v \in V} d(u,v)$.

For a connected graph $G = (V,E)$ and a source $s \in V(G)$,
$b(G,s)$ denotes the minimum number of rounds for broadcasting a message from $s$ to all the other vertices.

A \emph{broadcast protocol}, which represents how a message is broadcast from
the originator $s$, is a spanning tree $T$ rooted at $s$ together with a
time-stamp function $\tau \colon E(T)\to [t]$, informally indicating when an
edge is used to transmit a message. More formally, we require that $\tau$ must
be a proper coloring of the edges of $T$ with the additional property that for
each $v\in V\setminus \{s\}$ the edge connecting $v$ to its parent must have
minimum value assigned by $\tau$ over all edges incident on $v$. We say that
$u$ forwards the message to $v$ at time $t_{uv}$ if $u$ is the parent of $v$ in
the tree and $\tau(uv)=t_{uv}$.

Finally, we give a formal definition of {\TB} as follows.

\problemdef{\TB}
{Connected graph $G = (V,E)$, $s \in V$, and $t \in \mathbb{Z}^+$.}
{Determine whether $b(G,s) \leq t$.}

\section{NP-hardness Results}\label{sec:np_hardness}

In this section we prove that {\TB} remains NP-complete even on very restricted families of graphs,
namely on (i) cactus graphs that are distance-$1$ to path forest (\cref{thm:hardness:cactus})
and (ii) graphs of constant tree-depth (\cref{thm:hardness:treedepth}).
Our work builds upon previous work by Tale~\cite{tcs/Tale25},
who showed that the problem is NP-complete on graphs of distance-$2$ to path forest.
We note that \cref{thm:hardness:cactus} improves and significantly simplifies both the initial result by Tale~\cite{tcs/Tale25} as well as
a recent result by Aminian et al.~\cite{icalp/AminianKJS25} who,
via a long series of reductions, proved that {\TB} remains NP-complete on a special type of cactus graphs coined \emph{snowflake graphs}.%
\footnote{Roughly speaking, snowflake graphs are those that are distance-1 to caterpillar forest, so they are a strict superclass of graphs which are distance-$1$ from path forests.}

Since {\TB} clearly belongs to NP, it suffices to show the NP-hardness for each case.
The starting point of both of our reductions is a restricted version of \NMTS.
Given three multisets $A,B,C$ each containing $n$ positive integers such that $\sum_{a \in A} a + \sum_{b \in B} b = \sum_{c \in C} c$,
{\NMTS} asks to determine whether it is possible to partition the $3n$ given numbers into $n$ triples $S_1, \ldots, S_n$
such that each triple contains one $c_i \in C$, one $a_j \in A$, and one $b_k \in B$,
with $c_i = a_j + b_k$.
{\NMTS} is known to be strongly NP-complete~\cite[SP17]{fm/GareyJ79},
even when all input numbers are distinct~\cite[Appendix]{orl/HulettWW08},
that is, when $A \cup B \cup C$ is a set of cardinality $3n$.
Starting from such an instance,
by first doubling all elements of $A \cup B \cup C$ and then adding $1$ to those of $B \cup C$,
one can show that the following restricted version of the problem is NP-complete as well.

\problemdef{\RNMTS}
{Disjoint sets $A, B, C$, each containing $n$ positive integers in unary,
such that $\sum_{a \in A} a + \sum_{b \in B} b = \sum_{c \in C} c$,
and the elements of $A$ are even, while those of $B \cup C$ are odd.}
{Determine whether it is possible to partition the $3n$ given numbers into $n$ triples $S_1, \ldots, S_n$,
such that each triple contains one $c_i \in C$, one $a_j \in A$, and one $b_k \in B$,
with $c_i = a_j + b_k$.}

The ideas behind both reductions of \cref{thm:hardness:cactus,thm:hardness:treedepth} are quite similar.
We start by introducing the source vertex $s$, and setting $t$ to be the maximum target sum.
Then, for every target sum, we introduce a gadget on roughly that many vertices and connect both of its endpoints to $s$.
On a high level, if the target sum $c_i$ is due to $a_j + b_k$,
at rounds $t - a_j$ and $t - b_k$ the source $s$ is supposed to broadcast towards the said gadget;
notice that it is therefore crucial that $A \cap B = \varnothing$.
To restrict the broadcasting of $s$ during the rest of the rounds,
namely on round $h \in [t] \setminus \setdef{t - x}{x \in A \cup B}$,
we attach either a long path or a large star to $s$ that essentially forces it to broadcast towards it at round $h$.
Implementing the previous with the gadgets being long paths is sufficient for \cref{thm:hardness:cactus} to go through.
For \cref{thm:hardness:treedepth} however, we need to proceed in two steps and first reduce to the natural generalization of {\TB}
where there are multiple source vertices before eventually reducing to \TB.

\subsection{Cactus Graphs}
\label{ssec:cactus}

We start with the hardness result for cactus graphs that are distance-$1$ to path forest.

\begin{theorem}\label{thm:hardness:cactus}
    {\TB} is NP-complete for cactus graphs that are distance-$1$ to path forest.
\end{theorem}

\begin{proof}
    Let $\mathcal{I} = (A,B,C)$ be an instance of \RNMTS,
    where $A = \setdef{a_i}{i \in [n]}$, $B = \setdef{b_i}{i \in [n]}$, and $C = \setdef{c_i}{i \in [n]}$.
    We will construct an equivalent instance $\mathcal{J} = (G,s,t)$ of \TB,
    where $t = \max_{c \in C} c$.
    Note that $t > \max_{x \in A \cup B} x$.

    We construct the graph $G$ as follows:
    \begin{itemize}
        \item Introduce a vertex $s$;
        \item For every $i \in [n]$, introduce a path $\mathcal{P}_{i}$ on $c_i + 2$ vertices with endpoints
        $u_i$ and $v_i$, and add edges $\braces{s, u_i}$ and $\braces{s, v_i}$;
        \item For every $h \in [t] \setminus \setdef{t - x}{x \in A \cup B}$,
        introduce a path $\mathcal{Q}_{h}$ on $t - h + 1$ vertices,
        and add an edge from one of its endpoints to $s$.
        Name the other endpoint $q_{h}$ (for $h = t$,
        we name the single vertex of $\mathcal{Q}_t$ as $q_t$ instead).
    \end{itemize}
    This concludes the construction of $G$.
    Notice that $G$ is a cactus graph, while $G-s$ is a collection of paths.
    Before proving the equivalence of the two instances, we first compute the size of graph $G$.

    \begin{claim}\label{claim:hardness_path_forest:total_size}
        It holds that $|V(G)| = 1 + \frac{t(t+1)}{2}$.
    \end{claim}

    \begin{claimproof}
        Notice that $G$ is comprised of vertex $s$, as well as the vertices present in $\mathcal{P}_i$ and $\mathcal{Q}_{h}$, for $i \in [n]$ and $h \in [t] \setminus \setdef{t - x}{x \in A \cup B}$.
        Consequently,
        \begin{align*}
            |V(G)| &= 1 + \sum_{i=1}^n (c_i+2) + \sum_{h \in [t] \setminus \setdef{t - x}{x \in A \cup B}} (t - h+1)\\
            &= 1 + 2n + \sum_{i=1}^n c_i + (t-2n) + \sum_{h \in [t] \setminus \setdef{t - x}{x \in A \cup B}} (t - h)\\
            &= 1 + t + \sum_{i=1}^n c_i + \sum_{h \in [t]} (t - h) - \sum_{h \in \setdef{t - x}{x \in A \cup B}} (t - h)\\
            &= 1 + \sum_{i=1}^n c_i + \sum_{h \in [t]} h - \sum_{x \in A \cup B} x\\
            &= 1 + \frac{t(t+1)}{2},
        \end{align*}
        and the statement follows.
    \end{claimproof}

    \begin{claim}\label{claim:hardness_path_forest:rnmts->tb}
        If $\mathcal{I}$ is a yes-instance of {\RNMTS},
        then $\mathcal{J}$ is a yes-instance of {\TB}.
    \end{claim}

    \begin{claimproof}
        Let $(S_1, \ldots, S_n)$ be a partition of $A \cup B \cup C$ that certifies that $\mathcal{I}$ is a yes-instance.
        Rename appropriately the elements of $A \cup B \cup C$ so that $S_i = \{a_i,b_i,c_i\}$ for all $i \in [n]$.        
        Let $h \in [t]$ denote the current round, and consider the following broadcast protocol.
        \begin{itemize}
            \item If $h \in [t] \setminus \setdef{t - x}{x \in A \cup B}$,
            then set $s$ to forward the message along the path $\mathcal{Q}_{h}$.
            In the subsequent rounds, the path from $s$ to $q_{h}$ keeps propelling the message towards
            $q_{h}$.
            
            \item Otherwise, it holds that $h = t - x$ for some $x \in A \cup B$.
            If $x = a_i$, then set $s$ to forward the message to $u_i$, else if $x=b_i$,
            then towards $v_i$ instead.
            In the subsequent rounds, the message keeps getting propelled along the path $\mathcal{P}_i$
            towards its other endpoint.
        \end{itemize}
        This completes the description of the protocol.

        We now prove that every vertex of $G$ gets the message by the end of round $t$.
        Regarding the vertices belonging to $\mathcal{Q}_{h}$,
        notice that since $s$ forwards the message at round $h$,
        and $\mathcal{Q}_{h}$ contains $t - h + 1$ vertices,
        all vertices of $\mathcal{Q}_{h}$ receive the message by the end of round $t$;
        as a matter of fact, $q_{h}$ receives the message exactly on the last round $t$.

        It remains to argue for the vertices of $\mathcal{P}_{i}$.
        By the protocol, if $S_i = \braces{a_i, b_i, c_i}$, one endpoint of $\mathcal{P}_i$ gets the message at round
        $t - a_i$, while the other at round $t - b_i$.
        These two vertices forward the message towards each other to convey it to all the vertices in the path 
        connecting them.
        In that case, the number of vertices of $\mathcal{P}_i$ that receive the message by round $t$ is
        at least $t - (t - a_i) + 1 + t - (t - b_i) + 1 = a_i + b_i + 2 = c_i + 2$,
        thus all of its vertices receive the message by the end of round~$t$.
    \end{claimproof}

    \begin{claim}\label{claim:hardness_path_forest:tb->rnmts}
        If $\mathcal{J}$ is a yes-instance of {\TB},
        then $\mathcal{I}$ is a yes-instance of {\RNMTS}.
    \end{claim}

    \begin{claimproof}
        Assume there exists a broadcast protocol such that every vertex of $G$ receives the message by
        the end of round $t$.
        We claim that the maximum number of vertices that have received the message after $h \in [t]$ rounds
        of transmission is $1 + \sum_{i=1}^{h} i = 1 + \frac{h(h+1)}{2}$.
        To see this, notice that every vertex apart from $s$ has degree at most $2$, thus it can send the message to at most one other vertex.%
        \footnote{We assume that no vertex sends the message towards its already informed neighbors.}
        In that case, the number of vertices receiving the message at round~$h$ is at most those that received the message at round $h-1$ plus $1$ due to $s$.

        Since $|V(G)| = 1 + \frac{t(t+1)}{2}$ due to \cref{claim:hardness_path_forest:total_size},
        it follows that this is indeed the case,
        thus every vertex apart from $s$ and those that are informed in the last round broadcasts
        the message exactly once. 
        Consequently, every vertex of degree $1$ receives the message at the last round,
        thus $s$ broadcasts the message towards the path $\mathcal{Q}_{h}$ at round $h$,
        for all $h \in [t] \setminus \setdef{t - x}{x \in A \cup B}$.

        At round $h \in \setdef{t - x}{x \in A \cup B}$, $s$ broadcasts the message towards some path $\mathcal{P}_{i}$, and this broadcast results in a total of $x+1$ vertices of path $\mathcal{P}_i$ having received the message by round $t$.
        Consequently, it follows that for every path $\mathcal{P}_i$, $s$ broadcasts towards $u_i$ and $v_i$ at rounds
        $t - x^1_i$ and $t - x^2_i$, where
        $t - (t - x^1_i) + 1 + t - (t - x^2_i) + 1 = x^1_i + x^2_i + 2 = c_i + 2$.
        Since $c_i$ is odd, that is also the case for exactly one among $x^1_i$ and $x^2_i$,
        thus $\braces{x^1_i, x^2_i}$ contains an element from both $A$ and $B$.
        In that case, $(S_1, \ldots, S_n)$ certifies that $\mathcal{I}$ is a yes-instance,
        since it is a partition of $A \cup B \cup C$ where $S_i = \braces{c_i, x^1_i, x^2_i}$,
        each $S_i$ contains a single element out of $A$, $B$, and $C$, while $c_i = x^1_i + x^2_i$.
        This completes the proof.
    \end{claimproof}
    By \cref{claim:hardness_path_forest:rnmts->tb,claim:hardness_path_forest:tb->rnmts}, \cref{thm:hardness:cactus} follows.
\end{proof}

\subsection{Bounded Tree-depth}
\label{ssec:treedepth}

Let $G$ be a connected graph. An \emph{elimination tree} $T$ of $G$ is a rooted tree on the same vertex set, where every pair of vertices adjacent in $G$ adheres to the ancestor/descendant relation in $T$. The \emph{tree-depth} of $G$, denoted by $\td(G)$, is the minimum depth over all elimination trees of $G$. Let $G$ be a general (possibly disconnected) graph, an \emph{elimination forest} of $G$ is a set of elimination trees, one for each connected component of $G$. The tree-depth of $G$ is the minimum over all its elimination forests of the maximum depth of an elimination tree. Alternatively, it is the maximum tree-depth over its connected components.

We will use two basic facts about tree-depth. First, the tree-depth of a star is $2$; indeed, the star itself, with the central vertex as root, is a valid elimination tree of depth $2$.
Take a connected graph $G=(V,E)$, and $v\in V$. Then, $\td(G)\leq \td(G[V\setminus \{v\}])+1$; indeed, take an elimination forest of minimum depth of $G[V\setminus \{v\}]$, and add $v$ as a root, with the roots of the forest as children. It is an elimination tree of $G$, of depth $\td(G[V\setminus \{v\}])+1$.

Now we show the hardness for graphs of tree-depth at most~$6$.
To this end, we use a generalization of {\TB} as an intermediate problem,
where instead of a unique vertex as the initial source, we have a set of vertices (see, e.g.,~\cite{dam/JakobyRS98}).
We call the variant {\TBMS} and formalize it as follows, where $b(G,S)$ is defined analogously to $b(G,s)$.

\problemdef{\TBMS}
{Connected graph $G = (V,E)$, $S \subseteq V$, and $t \in \mathbb{Z}^+$.}
{Determine whether $b(G,S) \leq t$.}

Our proof thus consists of the following two steps.
First, we reduce {\RNMTS} to {\TBMS}.
On a high level, the reduction is similar to the one of \cref{thm:hardness:cactus},
albeit with a few supplementary steps to bound the tree-depth.
The main difference is that we use star-like structures equipped with multiple sources instead of paths.
Next, we present a reduction from {\TBMS} to {\TB} that does not increase tree-depth too much;
as a matter of fact, if this reduction is applied to a graph output by the first step,
then the tree-depth does not increase at all.
We again use star-like structures here.

We highlight that instead of independently presenting the two reductions in distinct theorems,
we choose to present them in one, as doing so allows us to present a tighter bound on the tree-depth of the final graph.

\begin{theoremrep}\label{thm:hardness:treedepth}
    {\TB} is NP-complete for graphs of tree-depth at most~$6$.
\end{theoremrep}

\begin{proof}
    We follow the two steps described above.

    \proofsubparagraph*{Reduction from {\RNMTS} to \TBMS.}

    Let $\mathcal{I} = (A,B,C)$ be an instance of \RNMTS,
    where $A = \setdef{a_i}{i \in [n]}$, $B = \setdef{b_i}{i \in [n]}$, and $C = \setdef{c_i}{i \in [n]}$.
    From $\mathcal{I}$, we construct an equivalent instance $\mathcal{J} = (G,S,t)$ of \TBMS.
    Set $t = \max_{c \in C} c$. Note that $t > \max_{x \in A \cup B} x$.
    We construct the graph $G$ as follows (see \cref{fig:reduction_treedepth_G,fig:reduction_treedepth_Pi}):
    \begin{itemize}
        \item introduce a vertex $s$;
        \item for $h \in [t] \setminus \setdef{t-x}{x \in A \cup B}$, introduce a star with center $v_h$ and $t-h$ leaves, and the edge $\{s,v_h\}$;
        \item for $i \in [n]$, introduce a gadget $\mathcal{P}_{i}$ composed of the following vertices and edges:
        \begin{itemize}
            \item two vertices $l_i$ and $r_i$ adjacent to $s$;
            \item for $j \in [c_{i}]$, vertices $l_{i,j}$, $s_{i,j}$, $r_{i,j}$ and
    	edges $\{l_i,l_{i,j}\}$, $\{l_{i,j},s_{i,j}\}$,  $\{s_{i,j}, r_{i,j}\}$, $\{r_{i,j}, r_i\}$;
            \item for $j \in [c_{i}]$ and $k \in [t-1]$, a star with a center $v_{i,j,k}$ and $t-k$ leaves, and the edge $\{s_{i,j},v_{i,j,k}\}$.
        \end{itemize}
    \end{itemize}
    We set $S = \{s\} \cup \setdef{s_{i,j}}{i \in [n], \, j \in [c_{i}]}$.
    Note that $|S| = 1 + \sum_{i=1}^{n} c_{i}$.
    Prior to proving the equivalence of instances $\mathcal{I}$ and $\mathcal{J}$ in \cref{lem:rnmts->tbms,lem:tbms->rnmts},
    we first show in \cref{lem:hardness_td1:total_size} that $|V(G)| = |S|(1 + t(t+1)/2)$,
    by a calculation similar to the one for \cref{claim:hardness_path_forest:total_size}.

    \begin{figure}[tbh]
        \centering
        \includegraphics{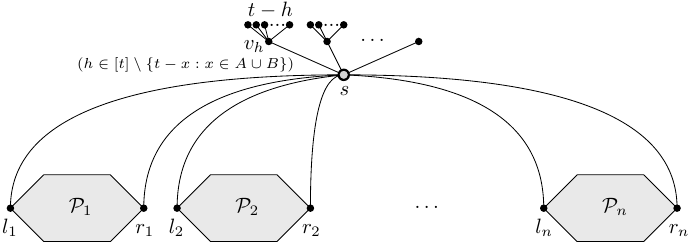}
        \caption{The graph $G$ in the reduction from \RNMTS\ to \TBMS.
        The initial sources are $s$ and all vertices $s_{i,j}$.
        For every $h \in [t] \setminus \setdef{t-x}{x \in A \cup B}$, the vertex $v_h$ is the center of a star with $t-h$ leaves;
        this forces $s$ to inform $v_h$ at round $h$, leaving the remaining rounds for informing the endpoints $l_i,r_i$ of the gadgets $\mathcal{P}_i$.}
        \label{fig:reduction_treedepth_G}
    \end{figure}
    \begin{figure}[tbh]
        \centering
        \includegraphics{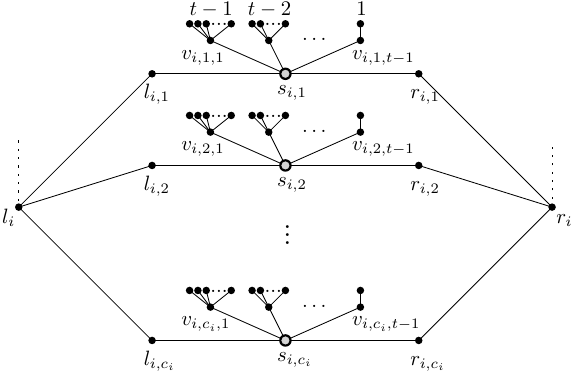}
        \caption{The target-sum gadget $\mathcal{P}_{i}$ associated with $c_i$.
        Each $s_{i,j}$ is an initial source and, for every $k \in [t-1]$, is adjacent to the center $v_{i,j,k}$ of a star with $t-k$ leaves.
        These stars force $s_{i,j}$ to inform their centers in rounds $1,\ldots,t-1$, leaving only round $t$ to inform one of $l_{i,j}$ and $r_{i,j}$;
        the transmissions through $l_i$ and $r_i$ then encode two values whose sum is $c_i$.}
        \label{fig:reduction_treedepth_Pi}
    \end{figure}

    \begin{lemma}\label{lem:hardness_td1:total_size}
        It holds that $|V(G)| = |S|(1 + t(t+1)/2)$.
    \end{lemma}

    \begin{nestedproof}
        Notice that $G$ is comprised of vertex $s$, the stars attached to it,
        as well as the vertices present in $\mathcal{P}_i$, for $i \in [n]$.
        The vertices present in $\mathcal{P}_i$ are exactly
        \[
            2 + c_i \cdot \parens*{3 + \sum_{j=2}^t j} =
            2 + c_i \cdot \parens*{2 + \frac{t (t+1)}{2}},
        \]
        consequently it follows that
        \begin{align*}
            |V(G)| &= 1 + \sum_{i=1}^n \parens*{2 + c_i \cdot \parens*{2 + \frac{t (t+1)}{2}}} + \sum_{h \in [t] \setminus \setdef{t - x}{x \in A \cup B}} (t - h + 1)\\
            &= 1 + 2n
            + \parens*{2 + \frac{t (t+1)}{2}} \cdot \parens*{\sum_{i=1}^n c_i}
            + (t-2n) + \sum_{h \in [t] \setminus \setdef{t - x}{x \in A \cup B}} (t - h)\\
            &= 1 + t + \parens*{2 + \frac{t (t+1)}{2}} \cdot \parens*{\sum_{i=1}^n c_i} + \sum_{h \in [t]} (t - h) - \sum_{h \in \setdef{t - x}{x \in A \cup B}} (t - h)\\
            &= 1 + \parens*{2 + \frac{t (t+1)}{2}} \cdot \parens*{\sum_{i=1}^n c_i} + \sum_{h \in [t]} h -
            \sum_{x \in A \cup B} x\\
            &= \parens*{1 + \frac{t (t+1)}{2}} \cdot \parens*{\sum_{i=1}^n c_i} + 1 + \frac{t (t+1)}{2}\\
            &= |S| \cdot \parens*{1 + \frac{t(t+1)}{2}}. \qedhere
        \end{align*}
    \end{nestedproof}

    \begin{lemma}\label{lem:rnmts->tbms}
        If $\mathcal{I}$ is a yes-instance of \RNMTS, then $\mathcal{J}$ is a yes-instance of \TBMS.
    \end{lemma}

    \begin{nestedproof}
    Since $\mathcal{I}$ is a yes-instance of \RNMTS,
    we can rename the elements of $A$ and $B$ so that $c_i=a_i+b_i$ for every $i \in [n]$.
    We construct a broadcast protocol for $G$ as follows:
    \begin{itemize}
        \item for $h \in [t]$,
        \begin{itemize}
            \item if $h \in [t] \setminus \setdef{t-x}{x \in A \cup B}$, $s$ forwards the message to $v_h$ at round $h$;
            \item if $h=t-a_i$, $s$ forwards the message to $l_i$ at round $h$;
            \item if $h=t-b_i$, $s$ forwards the message to $r_i$ at round $h$;
        \end{itemize}
        \item for $h \in [t] \setminus \setdef{t-x}{x \in A \cup B}$ and $h'\in [h+1,t]$, $v_h$ forwards the message to one of its $t-h$ leaves at round $h'$;
        \item for $i \in [n]$ and $j \in [a_i]$, $l_i$ forwards the message to $l_{i, j}$ at round $t-a_i +j$;
        \item for $i \in [n]$ and $j \in [b_i]$, $r_i$ forwards the message to $r_{i, a_{i}+j}$ at round $t-b_i +j$;
        \item for $i \in [n]$, $j\in [c_i]$, and $k\in [t]$
        \begin{itemize}
            \item if $k\in [t]$, $s_{i,j}$ forwards the message to $v_{i,j,k}$ at round $k$;
            \item if $k=t$ and $j\in [a_{i}]$, $s_{i,j}$ forwards the message to $r_{i,j}$ at round $k$;
            \item if $k=t$ and $j\in [a_i+1,a_i+b_{i}]$, $s_{i,j}$ forwards the message to $l_{i,j}$ at round $k$.
        \end{itemize}
        \item for $i \in [n]$, $j\in [c_i]$, $k\in [t-1]$, and $h\in [k+1,t]$, $v_{i,j,k}$ forwards the message to one of its $t-k$ leaves at time $h$;
    \end{itemize}

    We can see that this is a broadcast protocol for $(G,S)$, showing that $b(G,S) \le t$: each center of a star with $k$ leaves is broadcasted to at time $t-k$, and so it has exactly the time to broadcast to each of its leaves; furthermore, because $A,B,C$ is a yes-instance of \RNMTS, all vertices $l_{i,j}$ and $r_{i,j}$ are broadcasted to at time $t$.
    \end{nestedproof}

    \begin{lemma}\label{lem:tbms->rnmts}
        If $\mathcal{J}$ is a yes-instance of \TBMS, then $\mathcal{I}$ is a yes-instance of \RNMTS.
    \end{lemma}

    \begin{nestedproof}
        Throughout this proof, we fix a broadcast protocol of $(G,S)$ with at most $t$ rounds.

        For $i \in [n]$, $j \in [c_i]$, $k \in [t-1]$,
        $v_{i,j,k}$ has $t-k$ leaves and thus $s_{i,j}$ has to forward the message to $v_{i,j,k}$ by round $k$.
        This implies that $s_{i,j}$ forwards the message to $v_{i,j,k}$ exactly at round~$k$ for $k \in [t-1]$.
        Thus, only at the final round $t$ can $s_{i,j}$ forward the message to either $l_{i,j}$ or~$r_{i,j}$.

        We can see that a vertex informed by a non-source vertex never forwards the message as follows.
        It is clear for the degree-$1$ vertices.
        The centers of the stars attached to the sources receive the message directly from the sources.
        Clearly, $l_{i,j}$ never forwards the message to the source~$s_{i,j}$.
        If $l_{i,j}$ forwards the message to $l_{i}$, it has to receive the message from $s_{i,j}$ by round $t-1$, which is impossible as discussed above.
        Similarly, we can see that $r_{i,j}$ never forwards the message.
        This implies that $l_{i}$ and $r_{i}$ receive the message from $s$.
        Thus the claim holds.

        The claim in the previous paragraph implies that
        the number of informed vertices that forward the message increases by at most $|S|$ in each round:
        in the first round, there are at most $|S|$ such vertices;
        in the second round, there are at most $2|S|$ such vertices;
        and in general, there are at most $i |S|$ such vertices in the $i$-th round.
        Since there are $|S|$ informed vertices at the beginning
        and at most $i |S|$ vertices are newly informed in round $i$,
        there are at most $|S| + \sum_{i = 1}^{t} i |S| = |S| (1 + t(t+1)/2)$ informed vertices after $t$ rounds.

        Indeed, since $|V(G)| = |S|(1 + t(t+1)/2)$ by \cref{lem:hardness_td1:total_size},
        we need exactly $i |S|$ vertices that forward the message in each round $i$,
        and thus the source vertices and the vertices directly informed by them have to forward the message at every round after they are informed.
        In particular, each $v_{h}$ has to receive the message exactly at round~$h$.
        Hence, for $h \in [t] \setminus \setdef{t-x}{x \in A \cup B}$, $s$~forwards the message to $v_{h}$ at round $h$.
        This implies that for each $i \in [n]$, the rounds that $l_i$ and $r_i$ receive the message from $s$ belong to $\setdef{t-x}{x \in A\cup B}$.

        For $i\in [n]$, let $x_i^l, x_i^r \in A\cup B$ such that $l_i$ and $r_i$ receive the message at rounds $t-x_i^l$ and $t-x_i^r$, respectively.
        After $l_{i}$ is informed, it forwards the message to $x_{i}^{l}$ vertices in $\setdef{l_{i,j}}{j \in [c_{i}]}$.
        Similarly, $r_{i}$ forwards the message to $x_{i}^{r}$ vertices in $\setdef{r_{i,j}}{j \in [c_{i}]}$.
        Furthermore, for $j \in [c_{i}]$, $s_{i,j}$ forwards the message to exactly one of $l_{i,j}$ and $r_{i,j}$ at round $t$.
        In total, $x_{i}^{l} + x_{i}^{r} + c_{i}$ messages are forwarded to the $2c_{i}$ vertices in $\setdef{l_{i,j}}{j \in [c_{i}]} \cup \setdef{r_{i,j}}{j \in [c_{i}]}$,
        which gives us $x_{i}^{l} + x_{i}^{r} = c_{i}$.
        As the elements of $A$ are even and the ones of $B$ and $C$ are odd,
        one of $x_{i}^{l}$ and $x_{i}^{r}$ belongs to $A$ (we call this one $a'_{i}$) and the other belongs to $B$ (we call this one $b'_{i}$).
        Now the triples $S_{i} = \{a'_{i}, b'_{i}, c_{i}\}$ for $i \in [n]$ certify that $\mathcal{I}$ is a yes-instance of \RNMTS.
    \end{nestedproof}

    \proofsubparagraph*{Reduction from {\TBMS} to {\TB}.}

    Let $\mathcal{J}=(G,S,t)$ be an instance of \TBMS, where $S = \{s_{0}, s_{1}, \dots, s_{k}\}$.
    From $\mathcal{J}$, we construct an instance $\mathcal{K}=(G',s',t')$ of {\TB} by setting $s'= s_{0}$ and $t'=t+k+1$,
    and constructing $G'$ as follows.
    We start with $G$, retaining all its vertices and edges, and add the following forcing gadget (see \cref{fig:reduction_treedepth_G-prime}):
    \begin{itemize}
      \item for $i \in [k]$, introduce the edge $\{s_{0}, s_{i}\}$;
      \item for $i \in [k]$ and $j \in [i+1,k+1]$, introduce a star with center $v_{i,j}$
      and $t+k-j+1$ leaves, and the edge $\{s_i,v_{i,j}\}$;
      \item add a new vertex $z$, $t$ leaves attached to it, and the edge $\{s_{0}, z\}$.
    \end{itemize}
    No other vertices or edges are added.
    In particular,
    $E(G'[V(G)]) = E(G) \cup \setdef{\{s_0,s_i\}}{i \in [k]}$,
    and every newly added vertex is adjacent only to a vertex in $S$ or to another newly added vertex.
    Consequently, no vertex in $V(G) \setminus S$ is adjacent to a newly added vertex.

    \begin{figure}[tbh]
      \centering
      \includegraphics{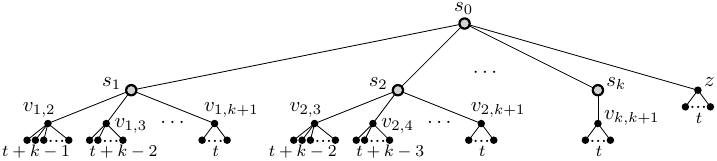}
      \caption{The forcing gadget added when reducing the instance $(G,S,t)$, with $S=\{s_0,\ldots,s_k\}$, from \TBMS\ to \TB; the original graph $G$ is not shown.
      The label below each star center indicates its number of leaves.
      The new stars force $s_0$ to inform $s_i$ at round $i$ and $z$ at round $k+1$, while every $s_i$ uses rounds $i+1,\ldots,k+1$ to inform its star centers.
      Hence, the sources use no original edge of $G$ before round $k+2$, after which they simulate the original $t$-round multiple-source instance.}
      \label{fig:reduction_treedepth_G-prime}
    \end{figure}

    The idea of this reduction is to force the source $s_{0}$ to forward the message to $s_{1}, \dots, s_{k}$ and $z$ in the first $k+1$ rounds, while preventing them to forward the message to any other vertices of $G$ in the meantime by forcing them to forward the message to the new stars. From round $k+2$ onwards, the vertices of $S$ can forward the message to the original vertices of $G$, simulating the instance of \TBMS.

    \begin{lemma}\label{lem:tbms->tb}
        If $\mathcal{J}$ is a yes-instance of \TBMS, then $\mathcal{K}$ is a yes-instance of \TB.
    \end{lemma}

    \begin{nestedproof}
        Let $P$ be a broadcast protocol for $(G,S)$ with at most $t$ rounds.
        We construct a broadcast protocol $P'$ for $(G',s_{0})$ as follows:
        \begin{itemize}
            \item for $i \in [k]$, $s_{0}$ forwards the message to $s_i$ at round $i$;
            \item $s_{0}$ forwards the message to $z$ at round $k+1$;
            \item for $i \in [k]$ and $j \in [i+1,k+1]$, $s_i$ forwards the message to $v_{i,j}$ at round $j$;
            \item for $i \in [k]$, $j \in [i+1,k+1]$, and $h \in [j+1,t+k+1]$,
            $v_{i,j}$ forwards the message to one of its $t+k-j+1$ leaves at round $h$;
            \item for $h \in [k+2, t+k+1]$, $z$ forwards the message to one of its $t$ leaves at round $h$;
            \item for any two vertices $v, v' \in V(G)$ and for every $h \in [k+2,t+k+1]$,
            $v$ forwards the message to $v'$ at round $h$ in $P'$ if and only if $v$ forwards the message to $v'$ at round $h-k-1$ in $P$.
        \end{itemize}

        We argue that every vertex of $G'$ has received the message by the end of round $t'$.
        At the end of round $k+1$, all vertices in $S \cup \{z\}$ and all centers of the stars attached to $S$ have been informed.
        Each center $v_{i,j}$ forwards the message to its $t+k-j+1$ leaves in rounds $j+1$ to $t+k+1$.
        Furthermore, $z$ forwards the message to its $t$ leaves in rounds $k+2$ to $t+k+1$.
        Finally, since $P$ forwards the message to each vertex of $V(G)$ by round $t$, $P'$ forwards the message to each vertex of $V(G)$ by round $t+k+1$.
    \end{nestedproof}

    \begin{lemma}\label{lem:tb->tbms}
        If $\mathcal{K}$ is a yes-instance of \TB, then $\mathcal{J}$ is a yes-instance of \TBMS.
    \end{lemma}

    \begin{nestedproof}
        Let $P$ be a broadcast protocol for $(G', s_{0})$ with at most $t+k+1$ rounds.
        By construction, no vertex in $V(G) \setminus S$ is adjacent to a new vertex in $G'$.
        Hence, it suffices to show that at the end of round $k+1$,
        the only vertices of $V(G)$ that have received the message are exactly those belonging to $S$,
        as we can use the rest of the rounds (i.e., rounds $k+2$ to $t+k+1$) of $P$ restricted to $G$ as a broadcast protocol for $(G,S)$ with at most $t$ rounds.

        For $i \in [k]$ and $j \in [i+1, k+1]$,
        the center $v_{i,j}$ has to be informed by round $j$
        since it has to forward the message to the $t+k-j+1$ leaves attached to it
        by the final round $t+k+1$.
        Let $i \in [k]$. Since $s_{i}$ has to forward the message to $v_{i,j}$ by round $j$ for all $j \in [i+1, k+1]$,
        $s_{i}$ has to be informed by round $i$.
        Furthermore, if $s_{i}$ is informed exactly at round $i$,
        then it has to forward the message to $v_{i,j}$ at round $j$ for each $j \in [i+1,k+1]$,
        and thus, $s_{i}$ cannot forward the message to any vertex in $V(G)$ during the first $k+1$ rounds.
        The discussion so far implies that, for each $i \in [k]$, $s_{0}$ forwards the message to $s_{i}$
        and $s_{i}$ never forwards the message to a vertex in $V(G)$ in the first $k+1$ rounds.
        At round $k+1$, $s_{0}$ has to forward the message to $z$ because $z$ has $t$ leaves attached to it.
    \end{nestedproof}

    \proofsubparagraph*{Bounding the tree-depth.}
    Let $\mathcal{I} = (A,B,C)$ be an instance of {\RNMTS},
    $\mathcal{J} = (G,S,t)$ the instance of {\TBMS} obtained from $\mathcal{I}$ by the first reduction,
    and $\mathcal{K} = (G',s',t')$ the instance of {\TB} obtained from $\mathcal{J}$ by the second reduction with $s_{0} = s$.
    By \cref{lem:tbms->tb,lem:tb->tbms,lem:rnmts->tbms,lem:tbms->rnmts},
    $\mathcal{I}$ is a yes-instance of {\RNMTS} if and only if $\mathcal{K}$ is a yes-instance of \TB.

    Now we show that the graph $G'$ is of tree-depth at most~$6$.
    Recall that $G'$ is obtained from $G$ by adding some stars attached to the vertices in $S$
    and edges $\{s_{0}, s_{i}\}$ for $i \in [k]$.
    Let $G''$ be the graph obtained from $G'$ by removing the vertex $s_{0}$.
    Each connected component of $G''$ is either a star
    or some $\mathcal{P}_{i}$ with additional stars attached to each $s_{i,j}$, which we call $\mathcal{P}'_{i}$.
    Let $\mathcal{P}''_{i}$ be the graph obtained from $\mathcal{P}'_{i}$ by removing the vertices $l_{i}$ and $r_{i}$.
    Each connected component of $\mathcal{P}''_{i}$ has vertices $l_{i,j}$, $s_{i,j}$, $r_{i,j}$ for some $j \in [c_{i}]$ and some stars attached to $s_{i,j}$.
    Removing $s_{i,j}$ from $\mathcal{P}''_{i}$ leaves only stars and singletons as connected components, and thus $\td(\mathcal{P}''_{i}) \le 3$.
    Therefore, we have $\td(G') \le \td(G'') + 1 \le  \td(\mathcal{P}'_{i}) + 1 \le \td(\mathcal{P}''_{i}) + 3 \le 6$.
\end{proof}

\section{Fixed-parameter Algorithms}\label{sec:fpt-alg}

\subsection{FPT Algorithm Parameterized by Vertex Integrity}\label{sec:vi}

Let $G = (V,E)$ be a graph.
The \emph{vertex integrity} of $G$, denoted $\vi(G)$, is the minimum integer $k$
such that there is a vertex set $S \subseteq V$ with $|S| + \max_{C \in \cc(G-S)} |V(C)| \le k$.
It is known that such a set can be computed in time $\bO(k^{k+1} n)$~\cite{algorithmica/DrangeDH16}.

The hardness for bounded tree-depth naturally raises the question of whether {\TB} becomes tractable
under the stronger restriction of bounded vertex integrity.
Indeed, $\td(G) \le \vi(G)$ for every graph $G$, so \cref{thm:hardness:treedepth}
does not rule out fixed-parameter tractability parameterized by vertex integrity.
The reductions in \cref{sec:np_hardness} do not settle this question either, since both constructions have unbounded vertex integrity:
for example, the construction in \cref{thm:hardness:cactus} contains induced paths $\mathcal{P}_i$ of unbounded order.
We show that {\TB} is fixed-parameter tractable for this parameter.

\begin{theorem}\label{thm:fpt:vi}
    {\TB} is fixed-parameter tractable parameterized by vertex integrity $k$,
    and can be solved in time $2^{\bO(k^5)} n^{\bO(1)}$.
\end{theorem}

\subparagraph{Overview.}

Our algorithm can be seen as a generalization of the fixed-parameter tractable algorithm parameterized by vertex cover number
due to Fomin et al.~\cite{tcs/FominFG24}.
Since the ideas are similar, let us present a brief overview of their algorithm.
Let $G = (V,E)$ be a graph with a vertex cover $S$ of size $k$.
They first show that there is an optimal broadcast protocol such that
all vertices in the vertex cover $S$ receive the message in the first $2k-1$ rounds.
Since a broadcast protocol with $2k-1$ rounds involves at most $2^{2k-1}$ vertices
and the vertices in $V \setminus S$ can be partitioned into $2^{k}$ classes of \emph{twins},
one can guess which vertices are involved in the first $2k-1$ rounds.
After testing the feasibility of the guessed vertices,
it only remains to check whether the completely informed vertex cover $S$
can send the message to the independent set $V \setminus S$ within a given number of rounds.
This problem can be solved in polynomial time by a reduction to the \textsc{Bipartite Matching} problem.

Our algorithm follows the same general strategy,
however there are two obstacles one needs to overcome in order to generalize from vertex cover number to vertex integrity.
Let $S \subseteq V$ be a set satisfying $|S|+ \max_{C \in \cc(G-S)} |V(C)| \le k$.
The first and main obstacle is showing the existence of an optimal broadcast protocol
in which the vertices in $S$ receive the message in the first few rounds.
Unfortunately, there are instances that do not admit such optimal broadcast protocols (see \cref{fig:vi-ce});
what we show instead is that there is an optimal broadcast protocol in which the vertices in $S$ receive
the message in the \emph{first} and \emph{last} few rounds.
The second obstacle is that after the removal of $S$, we have a set of small connected components instead of an independent set.
To deal with this, we show that most of the components of $G-S$ receive the message from $S$ only once and do not send the message back to $S$.
This property allows us to handle these components of $G-S$ as single vertices forming an independent set in some sense.

\begin{figure}[tbh]
  \centering
  \includegraphics{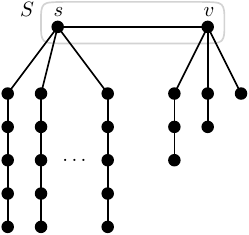}
  \caption{The set $S = \{s,v\}$ is the unique set satisfying that $|S| + \max_{C \in \cc(G-S)} |V(C)| \le 7$ ($= \vi(G)$).
  In all optimal broadcast protocols, $v$ is the last vertex that receives the message from~$s$
  as $v$ needs only three rounds after it receives the message, while other neighbors need four rounds.}
  \label{fig:vi-ce}
\end{figure}

\subparagraph{Some definitions.}

In the following, we fix a graph $G = (V,E)$ and a source vertex $s \in V$.
We also fix a vertex set $S \subseteq V$ that minimizes the sum $|S| + \max_{C \in \cc(G-S)} |V(C)|$ under the condition that $s \in S$.
Such $S$ can be found by first finding a set $S'$ minimizing $|S'| + \max_{C \in \cc((G - s)-S')} |V(C)|$, and then setting $S = S' \cup \{s\}$.
Let $k = |S| + \max_{C \in \cc(G-S)} |V(C)|$. Note that $k \le \vi(G) + 1$.

We say that $C$ and $C'$ in $\cc(G-S)$ have the same \emph{type}
if $G$ admits an automorphism $\eta$ such that $\eta(V(C)) = V(C')$, $\eta(V(C')) = V(C)$, and
$\eta(v) = v$ for each $v \in V \setminus (V(C) \cup V(C'))$.
By the equivalence relation of having the same type,
we partition $\cc(G-S)$ into the equivalence classes $\mathcal{C}_{1}, \dots, \mathcal{C}_{p}$.
We can upper-bound $p$, the number of types, by a function of $k$; e.g., $p < 2^{k^{2}}$,
which is an upper bound of the number of labeled graphs formed by at most $k$ vertices,
i.e., by the vertices in $S$ and a component $C \in \cc(G-S)$.
Observe that, in a broadcast protocol, two components of $G-S$ with the same type can be used interchangeably.


Let $P$ be a broadcast protocol for $(G,s)$.
We denote by $r(P)$ the number of rounds in $P$.
For $C \in \cc(G-S)$, let $a_{P}(C)$ denote the \emph{arrival time} at $C$ in $P$;
that is, $a_{P}(C)$ is the first round when some vertex in $S$ forwards the message to some vertex in $C$.
We call $C \in \cc(G-S)$ a \emph{leaf component} of $P$ if no vertex of $C$ forwards the message to vertices of $S$.
We call the vertices in the leaf components \emph{unimportant} and the rest of the vertices \emph{important}.
That is, the important vertices are those belonging to $S$ as well as in the non-leaf components of $G-S$.

\subsubsection{Optimal Broadcast Protocols with Nice Properties}

We say that an optimal broadcast protocol $P$ for $(G,s)$ is \emph{$S$-lazy}
if, among all optimal broadcast protocols, $P$ minimizes the total number of message-forwardings from~$S$.

\begin{lemmarep}\label{lem:vi-leaf-components}
If $P$ is an $S$-lazy optimal broadcast protocol for $(G,s)$,
then every leaf component $C$ with $a_{P}(C) \le r(P)-k+1$ receives the message from $S$ exactly once.
\end{lemmarep}

\begin{proof}
Let $P$ be an $S$-lazy optimal broadcast protocol.
Suppose to the contrary that $S$ forwards the message to some leaf component $C$ multiple times.
Let $P'$ be the broadcast protocol obtained from $P$ by removing
the message-forwardings between vertices in $C$ and
the message-forwardings from $S$ to $C$ except for an arbitrary one at round $a_{P}(C)$.
Observe that $P'$ broadcasts the message to $V \setminus V(C)$ by round $r(P)$
and to one vertex, say $v$, in $C$ at round $a_{P}(C)$.
Since $C$ is connected and has at most $k$ vertices, there is a broadcast protocol $P''$ for $(C, v)$ with at most $k-1$ rounds.
Now we can combine $P'$ and $P''$ by executing $P'$ for the first $a_{P}(C)$ steps
and then executing  $P''$ and the rest of $P'$ in parallel for the remaining $r(P) - a_{P}(C) \ge k-1$ steps.
The combined protocol, which we call $P'''$, broadcasts the message to all vertices.
This contradicts the assumptions on $P$
since $r(P''') \le r(P)$ and $P'''$ has a strictly smaller number of message-forwardings from $S$.
\end{proof}

\begin{lemmarep}
\label{lem:vi-important-vertices}
There is an $S$-lazy optimal broadcast protocol $P$ such that
all important vertices receive the message only in the first $(2k-3)k$ rounds and the last $2k$ rounds.
\end{lemmarep}
\begin{proof}
Let $P_{0}$ be an $S$-lazy optimal broadcast protocol for $(G,s)$
that satisfies the following two assumptions.
(1) No vertex takes a \emph{break} in $P_{0}$;
that is, if a vertex $v$ forwards the message to any neighbors,
then $v$ does so in the consecutive rounds starting at the round right after it receives the message.
(2) The vertices in $V \setminus S$ are \emph{greedy} in $P_{0}$;
that is, the following situation does not happen in $P_{0}$:
at some round $q$, an already informed vertex $u \in V \setminus S$ does not forward the message
although an uninformed neighbor $v \in N(u)$ does not receive the message.
It is easy to see that such an $S$-lazy optimal broadcast protocol exists.

We obtain $P$ from $P_{0}$ by exhaustively applying the following modification:
\begin{enumerate}
  \item find a vertex $u \in S$, an unimportant vertex $v$, and an important vertex $w$ such that
  $u$ forwards the message to $v$ at some round $q < r(P_{0})-k+1$ and to $w$ at round $q + 1$;

  \item swap the ordering of the message-forwardings from $u$ to $v$ and $w$,
  that is, force $u$ to forward the message to $w$ at round $q$ and to $v$ at round $q+1$;

  \item shift the subsequent message-forwardings starting at $v$ one round later and the ones starting at $w$ one round earlier.
\end{enumerate}

Observe that $P$ is an optimal protocol:
the message-forwardings starting at $v$ are restricted in a leaf component $C$,
and thus having $r(P_{0}) - (q+1) \ge k-1 \ge |V(C)|-1$ rounds suffice;
the subsequent message-forwardings starting at $w$ are done one round earlier than in $P_{0}$.
Note that $P$ is still $S$-lazy, no vertex takes a break in $P$, and the vertices in $V \setminus S$ are greedy.
The modification gives the following property of $P$.
\begin{claim}
\label{clm:vi-leaf-last}
If $u \in S$ forwards the message to an unimportant vertex at round $q < r(P) - k + 1$
and to some vertex $v$ at round $q+1$,
then $v$ is also an unimportant vertex.
\end{claim}

\begin{claimproof}
Observe that the previously described modification does not change the sets of important and unimportant vertices, as it does not change the set of vertices each vertex forwards the message to. Therefore, the sum of the reception times of all important vertices strictly decreases each time the modification is applied. Hence, the exhaustive application of the modification terminates.
\end{claimproof}

The greediness of the vertices in $V \setminus S$ implies the following property of $P$.
\begin{claim}
\label{clm:vi-nonlazy-components}
No vertex of $C \in \cc(G-S)$ forwards or receives the message at round $a_{P}(C) + k$ and later.
\end{claim}
\begin{claimproof}
Observe that since $|N[V(C)]| \le  |V(C) \cup S| \le k$,
the number of rounds in which at least one vertex in $C$ forwards or receives the message is at most $k$.
Suppose to the contrary that some vertex of $C$ forwards or receives the message at round $a_{P}(C) + k$ or later.
This implies that there is a round $q \in [a_{P}(C)+1, a_{P}(C) + k -1]$
in which no vertex of $C$ forwards the message
although there are informed vertices in $V(C)$ and uninformed vertices in $N[V(C)]$.
In particular, since $G[N[V(C)]]$ is connected, there are $u \in V(C)$ and $v \in N(u)$
such that $u$ receives the message before round $q$ but does not send the message at round $q$
and $v$ does not receive the message at round $q$ or before.
This contradicts the greediness assumption of the vertices in $V \setminus S$ ($\supseteq V(C)$).
\end{claimproof}

\begin{claim}
\label{clm:vi-S-first}
Let $q \in [2k-2, r(P)-k]$.
If no vertex forwards the message to $S$ in any of rounds $q-2k+3, \dots, q$,
then no vertex forwards the message to $S$ at round $q+1$.
\end{claim}
\begin{claimproof}
To prove the claim,
we first show that if a vertex in some component $C \in \cc(G-S)$ receives or forwards the message at round $q-k+1$, then $C$ is a leaf component.
From the definition of $a_{P}(C)$, we have $a_{P}(C) \le q-k+1$.
By \cref{clm:vi-nonlazy-components}, $q-k+1 < a_{P}(C) + k$ holds.
Combining these two bounds, we obtain $q-2k+2 \le a_{P}(C) \le q-k+1$.
By \cref{clm:vi-nonlazy-components} again,
the vertices in $C$ may forward the message in rounds $a_{P}(C)+1, \dots, a_{P}(C)+k-1$.
Therefore, a vertex in $C$ may forward the message only in rounds $q-2k+3, \dots, q$.
Since no vertex forwards the message to $S$ in these rounds, $C$ is a leaf component.

The discussion in the previous paragraph shows that non-leaf components do not receive the message at round $q-k+1$.
Since no vertex in $S$ receives the message at round $q-k+1$,
the vertices that receive the message at round $q-k+1$ are unimportant vertices.

We now show that no vertex in $S$ forwards the message to a vertex in $S$ at round $q+1$.
We actually show a stronger fact that no vertex in $S$ forwards the message to an important vertex at any round $q' \in [q-k+1, q+1]$.
Since no vertex in $S$ receives the message in rounds $q-k+1, \dots, q$,
the set of informed vertices in $S$ is the same at round $q-k+1$ and at the beginning of round~$q'$.
This implies that a vertex $v \in S$ forwards the message at round $q'$ only if it does so at round $q-k+1$
(as no vertex takes a break in $P$).
Since the vertices that receive the message at round $q-k+1$ are unimportant vertices,
\cref{clm:vi-leaf-last} implies that if a vertex in $S$ forwards the message at round $q'$, then the destination is an unimportant vertex.

Finally, we show that no vertex in a non-leaf component forwards the message to anywhere at round $q+1$.
If a vertex in a non-leaf component $C \in \cc(G-S)$ forwards the message at round $q+1$,
then $q+1 < a_{P}(C) + k$ holds by \cref{clm:vi-nonlazy-components},
and thus $q-k+2 \le a_{P}(C) \le q$.
This implies that some vertex in $S$ forwards the message to $C$ at round $q' \in [q-k+2, \dots, q]$,
which is impossible as shown in the previous paragraph.
\end{claimproof}

\begin{claim}
\label{clm:vi-S-first-last}
The vertices in $S$ receive the message only in the first $2(k-1)^{2}$ rounds and the last $k$ rounds.
\end{claim}
\begin{claimproof}
Let $q$ be the latest round such that $q \le r(P)-k$ and a vertex in $S$ receives the message at round~$q$.
(If no such round exists, then we are done.)
It suffices to show that $q \le 2(k-1)^2$.
By \cref{clm:vi-S-first},
the maximum number of consecutive rounds before round $q$ in which no vertex in $S$ receives the message is less than $2k-2$.
That is, in every interval of $2k-2$ consecutive rounds before round $q$,
at least one vertex in $S$ receives the message.
Since $|S \setminus \{s\}| \le k-1$, we have $q \le (2k-2)(k-1)$.
\end{claimproof}

Now, to complete the proof of the lemma,
it suffices to show that no vertex in a non-leaf component receives the message at any round in $[(2k-3)k + 1, r(P) - 2k]$.
Let $C \in \cc(G-S)$.
Assume that some vertex in $C$ receives the message at some round in $[(2k-3)k + 1, r(P) - 2k]$.
By \cref{clm:vi-nonlazy-components}, it holds that $a_{P}(C) \ge (2k-3)k + 1 -(k-1) = 2(k-1)^{2}$.
By \cref{clm:vi-nonlazy-components}, the vertices in $C$ may forward the message only in rounds
$[a_{P}(C)+1, a_{P}(C) + k-1] \subseteq [2(k-1)^{2} + 1, r(P) - k - 1]$.
By \cref{clm:vi-S-first-last}, no vertex in $S$ receives the message in these rounds, and thus $C$ is a leaf component.
\end{proof}

\subsubsection{The Algorithm Parameterized by Vertex Integrity}

Now we present a fixed-parameter tractable algorithm parameterized by vertex integrity for finding an optimal broadcast protocol for $(G,s)$.
We first guess the number of rounds $t$ in an optimal broadcast protocol. This is possible as $t < |V|$.
If $t < 2k^{2}$, then we apply the fixed-parameter algorithm of \cref{thm:fpt:twt} which runs in time $2^{\bO(k^3)} n^{\bO(1)}$
(recall that the vertex integrity of a graph upper-bounds its treewidth).
In the following, we assume that $t \ge 2k^{2}$.

We find an optimal broadcast protocol $P$ for $(G,s)$ with $t$ rounds such that
\begin{itemize}
  \item every leaf component receives the message from $S$ exactly once, and
  \item the important vertices receive the message only in the first $(2k-3)k$ rounds and the last $2k$ rounds.
\end{itemize}
Such a protocol exists by \cref{lem:vi-leaf-components,lem:vi-important-vertices}.
As in the proof of \cref{lem:vi-important-vertices}, we also assume that no vertex takes a break in $P$.

To describe the algorithm, we partition the set of rounds $[t]$ in $P$ into three subsets of consecutive rounds $R_{1} = [1, (2k-3)k]$, $R_{2} = [(2k-3)k + 1, t-2k]$, and $R_{3} = [t-2k+1, t]$.
Note that $|R_{2}| = (t-2k) - ((2k-3)k + 1) + 1 = t - 2k^{2}  +k \ge k$ as $t \ge 2k^{2}$.

The algorithm works as follows.
\begin{enumerate}\setcounter{enumi}{-1}
    \item Guess $S_{1} \subseteq S$ and $\mathcal{C}_{1}, \mathcal{C}_{3} \subseteq \cc(G-S)$ as follows:
    \begin{itemize}
        \item $s \in S_{1}$ and $S_{1} \setminus \{s\}$ is the set of vertices in $S$ that receive the message in $R_{1}$;

        \item $\mathcal{C}_{1} = \setdef{C \in \cc(G-S)}{a_{P}(C) \in R_{1}}$
        with $|\mathcal{C}_{1}| \le (2k-3)k^{2}$;

        \item $\mathcal{C}_{3} = \setdef{C \in \cc(G-S)}{a_{P}(C) \in R_{3}}$
        with  $|\mathcal{C}_{3}| \le 2k^{2}$.
    \end{itemize}

    \item Check if the combination of $S_{1}$ and $\mathcal{C}_{1}$ is feasible: that is,
    the algorithm checks whether there is a \emph{partial} broadcast protocol $P_{1}$
    for $G[S_{1} \cup \bigcup_{C \in \mathcal{C}_{1}} V(C)]$ within $(2k-3)k$ rounds
    that starts at $s$ and informs all vertices in $S_{1}$ and at least one vertex of each $C \in \mathcal{C}_{1}$.

    \item Let $\mathcal{C}_{2} = \cc(G-S) \setminus (\mathcal{C}_{1} \cup \mathcal{C}_{3})$.
    Check if $S_{1}$ can forward the message to each of $\mathcal{C}_{2}$ in~$R_{2}$: that is,
    the algorithm checks whether there is a \emph{multi-source partial} broadcast protocol~$P_{2}$
    for $G[S_{1} \cup \bigcup_{C \in \mathcal{C}_{2}} V(C)]$
    with $t - (2k-1)k$ rounds
    that starts at $S_{1}$ and informs exactly one vertex (equivalently, at least one vertex) of each $C \in \mathcal{C}_{2}$.

    \item Check if the combination of $S_{1}$ and $\mathcal{C}_{3}$ is feasible: that is,
    the algorithm checks whether there is a \emph{multi-source} broadcast protocol $P_{3}$
    for $G[S \cup \bigcup_{C \in \mathcal{C}_{3}} V(C)]$ with $2k$ rounds
    that starts at $S_{1}$.
\end{enumerate}

\subparagraph*{Correctness.}
We first show that an optimal broadcast protocol for $(G,s)$ with $t$ rounds exists if and only if
correctly guessed $S_{1}$, $\mathcal{C}_{1}$, $\mathcal{C}_{3}$ pass all three tests by the algorithm.

If an optimal broadcast protocol for $(G,s)$ with $t$ rounds exists, then there is one with the properties that we assumed in the algorithm.
Let $P$ be such a protocol. We construct $S_{1}$, $\mathcal{C}_{1}$, $\mathcal{C}_{3}$ from $P$ as defined in the algorithm.
Observe that at each round, at most $|S| \le k$ components in $\cc(G-S)$ may receive the message from $S$,
and thus $|\setdef{C \in \cc(G-S)}{a_{P}(C) \in R_{1}}| \le (2k-3)k^{2}$ holds.
By the same reasoning, we have $|\setdef{C \in \cc(G-S)}{a_{P}(C) \in R_{3}}| \le 2k^{2}$.
Clearly, $S_{1}$, $\mathcal{C}_{1}$, $\mathcal{C}_{3}$ pass all three tests.

Conversely, assume that some $S_{1}$, $\mathcal{C}_{1}$, $\mathcal{C}_{3}$ pass all three tests by the algorithm.
Let $P_{1}$, $P_{2}$, $P_{3}$ be the protocols found by the algorithm.
Let $P'$ be the protocol obtained by concatenating $P_{1}$, $P_{2}$, $P_{3}$ in this ordering.
In $P'$, some vertices of the components in $\mathcal{C}_{1}$ and $\mathcal{C}_{2}$ may remain uninformed,
while all other vertices receive the message.
We modify $P'$ to obtain a desired protocol.
We first consider a component $C_{1} \in \mathcal{C}_{1}$ that has uninformed vertices in $P'$.
In $P'$, at least one vertex in $C_{1}$ receives the message in $P_{1}$ and no vertex in $C_{1}$ sends or receives the message in $P_{2}$.
Thus, we can modify $P'$ in such a way that the vertices in $C_{1}$ send and receive the message locally in $C_{1}$
and inform all vertices of $C_{1}$ using at most $k-1$ rounds in $P_{2}$.
Next we consider a component $C_{2} \in \mathcal{C}_{2}$ that has uninformed vertices in $P'$.
Similarly to the previous case, we can modify $P'$ locally in $C_{2}$ to inform all vertices of $C_{2}$ using at most $k-1$ rounds in $P_{3}$.
Now the obtained protocol is a broadcast protocol for $(G,s)$ with $t$ rounds.

\subparagraph*{Implementation details and Running time.}
Now we expand upon some of the implementation details of our algorithm and argue that it runs in time $2^{\bO(k^5)} n^{\bO(1)}$.

We first estimate the number of possible candidates for guessing $S_{1}$, $\mathcal{C}_{1}$, and $\mathcal{C}_{3}$.
For $S_{1}$, we simply try all $2^{|S|-1} = 2^{k-1}$
subsets of $S$ that contain $s$.
For $\mathcal{C}_{i}$ ($i \in \{1,3\}$),
it suffices to try all possibilities of how many components of each type are included in $\mathcal{C}_{i}$.
Since the number of types is less than $2^{k^{2}}$,
the number of candidates for $\mathcal{C}_{i}$ can be bounded from above by a function in $k$;
i.e., $2^{k^{2} \cdot (2k-3)k^{2} }$ for $\mathcal{C}_{1}$ and $2^{k^{2} \cdot 2k^{2}}$ for $\mathcal{C}_{3}$.
In total, we test at most $2^{2k^{5}}$ candidates for the combination of $S_{1}$, $\mathcal{C}_{1}$, $\mathcal{C}_{3}$.
Thus, it suffices to show that, for a fixed combination of $S_{1}$, $\mathcal{C}_{1}$, $\mathcal{C}_{3}$,
each of the three tests can be done within time~$2^{\bO(k^5)} n^{\bO(1)}$.

We first argue about the first and the third test.
For the first test, we can exhaustively guess for each vertex of the
target graph $G[S_{1} \cup \bigcup_{C \in \mathcal{C}_{1}} V(C)]$ the round it receives the message within the first $(2k-3)k$ rounds,
and then verify feasibility in polynomial time; as this is a partial broadcast protocol, we can also consider an ``uninformed'' value in our guesses.
Since the target graph has at most $k + |\mathcal{C}_{1}| k < 2k^{4}$ vertices, we consider $k^{\bO(k^{4})}$ cases,
thus in total we can do the first test in time $2^{\bO(k^4 \log k)} n^{\bO(1)}$.
Using a similar approach, we can do the third test in time $2^{\bO(k^3 \log k)} n^{\bO(1)}$.

For the second test, we can use the polynomial-time subroutine that Fomin et al.~\cite{tcs/FominFG24} used for their algorithm parameterized by vertex cover number.
The subroutine checks in polynomial time whether there is a multi-source broadcast protocol with a given number of rounds starting at a vertex cover.
To use this subroutine, we construct a graph $G_{2}$ from $G[S_{1} \cup \bigcup_{C \in \mathcal{C}_{2}} V(C)]$
by replacing each component $C_{2}$ in $\mathcal{C}_{2}$ with a single vertex that is adjacent to the vertices in $S_{1} \cap N(V(C_{2}))$.
Observe that the desired multi-source partial broadcast protocol $P_{2}$ for $G[S_{1} \cup \bigcup_{C \in \mathcal{C}_{2}} V(C)]$ exists if and only if
there is a multi-source broadcast protocol for $G_{2}$ with $t-(2k-1)k$ rounds starting at $S_{1}$.
Since $S_{1}$ is a vertex cover of $G_{2}$, we can use the subroutine of Fomin et al.~\cite{tcs/FominFG24}.

\subsection{FPT Algorithm Parameterized by Distance to Clique}
\label{sec:dtc}

In this section we consider {\TB} parameterized by distance to clique $k$,
in which the input graph contains a subset $X\subseteq V$ of size $k$ such that $G-X$ is a clique.
We call $X$ a \emph{modulator} of $G$ and let $n = |V| -k$.
Computing a (smallest) modulator is equivalent to computing a (minimum) vertex cover and thus can be done in time $2^{k}n^{\bO(1)}$~\cite{tcs/ChenKX10,stacs/HarrisN24}.

We show that, in these instances, the amount of information necessary to encode a solution is bounded,
in that we can recreate a solution knowing the following about the vertices of the modulator:
on which round each is informed,
how many vertices each informs,
as well as some details about the vertices that inform them
(see \cref{lem:dtc:mechanism} and the description below for details).
We solve the problem by enumerating all possibilities for this information,
and then recreating a matching solution in polynomial time.

Our approach is based on the following two facts: %
(i) $b(G,s) \leq k + \lceil\log_2 n\rceil$, and %
(ii) to encode a solution we only need to know the information described above.

The rest of this section is dedicated to proving the following theorem.

\begin{theorem}\label{thm:fpt:dtc}
{\TB} is fixed-parameter tractable parameterized by distance to clique $k$, and can be solved in time $2^{\bO(k^2)} n^{\bO(1)}$.
\end{theorem}

We start by formalizing and proving the two facts mentioned above:

\begin{lemmarep}
\label{lem:dtc:bound}
Let $G=(V,E)$ be a connected graph and $X \subseteq V$ with $|X|=k$ such that $G-X$ is a clique of size~$n$.
Then $b(G,s) \leq k + \lceil\log_2 n\rceil$ for all $s \in V$.
\end{lemmarep}

\begin{proof}
We describe a broadcast protocol that informs every vertex over three distinct phases and within the stated time. %
We remark that vertices may remain idle in some rounds, even if they have uninformed neighbors.

The first phase of the strategy is to inform a vertex of the clique.
Let $\ell \geq 0$ be the minimum number of rounds needed to inform a vertex of the clique (closest to $s$),
and let $s' \in V \setminus X$ be that vertex (if $s \in V \setminus X$, then $\ell=0$ and  $s'=s$).
Since we choose $\ell$ to be minimum, at least $\ell$ vertices in $X$ are informed through this process, since the vertices informed before $s'$ are in $X$.

The second phase is to inform all the vertices in the clique, which takes $\lceil\log_2 n\rceil$ rounds, as each round doubles the number of informed vertices in $V\setminus X$.
Since $s' \in V\setminus X$ is informed at time $\ell$, by round $\lceil\log_2 n\rceil + \ell$ every vertex in the clique is informed.

The last phase is to inform the remaining at most $k-\ell$ vertices in $X$ in $k-\ell$ rounds, which can be done by informing one vertex per round using the fact that $G$ is connected.

In total, this strategy finishes in  $\ell + \lceil\log_2 n\rceil + (k-\ell) = k + \lceil\log_2 n\rceil$ rounds.
\end{proof}

The lemma below formalizes the process of recreating a solution, that is, given restricted information about a broadcast protocol, computing a protocol that matches that information.
For simplicity of presentation, the lemma lists in compact form the information that is needed about the broadcast protocol,
and we provide additional explanation below the lemma.

\begin{lemmarep}
\label{lem:dtc:mechanism}
Let $(G,s)$ be an instance of {\TB},
$X$ be a modulator of $G$ of size $k$, and
$P$ be a broadcast protocol.
We can construct in polynomial time a broadcast protocol $P'$ with the same number of rounds as $P$, given the following information about $P$:
\begin{enumerate}
	\item $t$, the number of rounds of $P$;
	\item for each $v \in X$, the round $\tau_v$ in which $v$ is informed;
	\item for each $v \in X$, the vertex $a_v$ that informs it (its \emph{informer}), where a vertex in $V\setminus (X\cup \{s\})$ is identified by its neighborhood in $X$ as well as the set of vertices in $X$ it informs;
	\label{lem:dtc:mechanism:3}
	\item for each informer $a_v$, $v \in X$, its informer $b_v$ if it is in $X$, or $b_v = \phi$ otherwise (representing either an informer in $V\setminus X$ or that $a_v = s$);
	\label{lem:dtc:mechanism:4}
	\item for each $v \in X$, the number $c_v$ of vertices that it informs,
	that is, how many times it transmits the message to an uninformed vertex.
	\label{lem:dtc:mechanism:5}
\end{enumerate}
\end{lemmarep}

For Point \ref{lem:dtc:mechanism:3}, we do not need to know the exact vertex of the clique that informs a vertex $v \in X$,
as there can be many clique vertices that are functionally equivalent.
Instead, we identify vertices in $V\setminus (X\cup \{s\})$ by their neighborhood in $X$;
if multiple vertices in $X$ share an informer, the first (in an arbitrary order) encodes it by its neighborhood, while the following simply indicate that they have the same informer as a previous vertex.
We can thus encode the informer of $u \in X$ using a number $i_u$ from $0$ to $2^k+2k$ as follows. %
Take an arbitrary ordering of $X = \{v_0, \ldots, v_{k-1}\}$. %
If $i_u < 2^k$, $a_u$ is a (yet unused) vertex with neighborhood in $X$ given by the binary representation of $i_u$;
if $2^k \leq i_u < 2^k+k$, $a_u$ is the vertex of $X$ with index $i_u-2^k$ (i.e.~$v_{i_u-2^k}$); %
if $2^k+k \leq i_u < 2^k+2k$, $a_u$ is the informer of the vertex of $X$ with index $i_u-2^k-k$, that is, $a_u = a_w$ with $w = v_{i_u-2^k-k}$;
and if $i_u = 2^k+2k$, then $a_u=s$.
For Point \ref{lem:dtc:mechanism:4}, we need to know the informer $b_v$ of each $a_v$;
however, if $b_v$ is in $V \setminus X$, the actual vertex is irrelevant, and thus we simply notate this case as $b_v = \phi$; we notate it similarly if $a_u=s$, in which case it has no informer.

\begin{proof}[Proof of \cref{lem:dtc:mechanism}]
The idea of the proof and its construction is that clique vertices are mostly interchangeable,
and thus once we know how and when each vertex in $X$ is informed, we can simply inform clique vertices whenever otherwise possible.
The two main details to consider are that an informer $a_v$ must be informed before round $\tau_v$,
and that each vertex in $X$ can only inform a subset of the vertices in $V\setminus X$,
and thus the vertices in $V\setminus X$ must be assigned appropriately.

We start with an empty protocol $P'=(T',\tau')$, where $T'=(V, \varnothing)$ is an empty tree, and then fill it in step-by-step. %
Let $A = \setdef{a_v}{v \in X} \setminus X$ be the set of clique vertices that inform vertices in $X$.
Let $R = V \setminus \{s\}$ be the set of yet uninformed vertices.

First, for each $a_v \in A$, we set $P'$ so that $a_v$ sends a message to $v$ at time $\tau_v$, and we remove $v$ from $R$. %
Formally, we add the edge $a_vv$ to $T'$ and set $\tau'(a_vv) = \tau_v$.
Note that, up to this point, we make no effort to have $P'$ connected,
in the sense that informers are not yet necessarily informed.

The next step is to make sure that informers are themselves informed,
by assigning them to $b_v$ if in $X$ or to a clique vertex if $b_v = \phi$.
This step also makes sure that the protocol is connected, that is, it is feasible when restricted to the vertices that have received a message.
Let $v \in X$ be a vertex with minimum $\tau_v$ such that $a_v \in R$.
Find the earliest round $i < \tau_v$ such that a vertex $u$ is both idle and informed before round $i$; $u$ must equal $b_v$ if $b_v \neq \phi$, but can be any vertex $u \in V\setminus X$ otherwise.
Set $P'$ so that $u$ informs $a_v$ on round $i$, by adding $ua_v$ to $T'$ with $\tau'(ua_v) = i$, and remove $a_v$ from $R$.
Then, repeat this process until no more vertices remain in $A \cap R$; we later show that this is always possible if $P$ is a feasible protocol.

All that is left to do now is to add to $P'$ the remaining vertices in $R$, which are all clique vertices.
As we mentioned before, it is important to assign these vertices correctly, as some vertices of the modulator may be limited in which clique vertices they can inform.
To do this, we solve a maximum flow problem to assign clique vertices to each $v \in X$; any leftover vertices are informed by other clique vertices.
As each vertex $v \in X$ informs $c_v$ many vertices, at this stage we assign it
$c'_v := c_v - |\setdef{u \in X}{a_u=v}| - |\setdef{a_u}{u \in X, b_u = v}|$ vertices,
which excludes the vertices that it is already set to inform.
Let $q \in \{0,1\}^X$ be the indicator vector of a subset of $X$, and let $R_q$ be the subset of vertices of $R$ whose neighborhood corresponds to $q$.
Our instance of maximum flow is a graph with four layers (see \cref{fig:dtc:max-flow} for an example):
\begin{enumerate}
	\item A layer with a single vertex $\mathbf s$,
	\item A layer with vertices $h_v$ for every $v \in X$, with incoming arcs $(\mathbf s, h_v)$ of capacity equal to the number of vertices to be informed, $c'_v$;
	\item A layer with vertices $r_q$ for every $q \in \{0,1\}^X$; there are arcs $(h_v, r_q)$ with unbounded capacity, for every $v \in X$ and $q \in \{0,1\}^X$ such that $q(v) = 1$;
	\item A layer with a single vertex $\mathbf t$, with incoming arcs $(r_q, \mathbf t)$ of capacity $|R_q|$.
\end{enumerate}

\begin{figure}
\centering

\begin{tikzpicture}[scale=.75]
\tikzstyle{nd}=[fill=black, draw=none, shape=circle]
\tikzstyle{edge}=[-, thick]

\node [nd, label={left:$\mathbf{s}$}] (0) at (0, 0) {};

\node [nd, label={above:$v \in X$}] (1) at (5, 3) {};
\node [nd] (2) at (5,  0) {};
\node [nd] (3) at (5, -3) {};

\node [nd] (5) at (10, 3.5) {};
\node [nd] (6) at (10, 2.5) {};
\node [nd] (7) at (10, 1.5) {};
\node [nd] (8) at (10, 0.5) {};
\node [nd] (9) at (10, -0.5) {};
\node [nd] (10) at (10, -1.5) {};
\node [nd] (11) at (10, -2.5) {};
\node [nd] (12) at (10, -3.5) {};

\node [nd, label={right:$\mathbf{t}$}] (15) at (15, 0) {};

\draw [edge] (0) to[edge label={$c_v$}] (1);
\draw [edge] (0) to (2);
\draw [edge] (0) to (3);

\draw [edge] (1) to (5);
\draw [edge] (1) to (6);
\draw [edge] (1) to (7);
\draw [edge] (1) to (8);
\draw [edge] (2) to (5);
\draw [edge] (2) to (6);
\draw [edge] (2) to (9);
\draw [edge] (2) to (10);
\draw [edge] (3) to (5);
\draw [edge] (3) to (7);
\draw [edge] (3) to (9);
\draw [edge] (3) to (11);

\draw [edge] (5) to[edge label={$|R_q|$}] (15);
\draw [edge] (6) to (15);
\draw [edge] (7) to (15);
\draw [edge] (8) to (15);
\draw [edge] (9) to (15);
\draw [edge] (10) to (15);
\draw [edge] (11) to (15);
\draw [edge] (12) to (15);

\draw[very thick, dashed] (10,2) ellipse [x radius=1cm, y radius=1.8cm]; 
\node at (11.8,3.7) {$q : q(v)=1$};
\end{tikzpicture}
\caption{Example max-flow instance for $|X|=3$. Each vertex has $c'_v$ free rounds to inform clique vertices in its neighborhood, and a max-$\mathbf{s}$-$\mathbf{t}$-flow indicates how to assign clique vertices to $X$ so as to maximize the number of clique vertices informed by $X$.}
\label{fig:dtc:max-flow}
\end{figure}

To obtain the assignment of clique vertices to vertices in $X$, it suffices to obtain a maximum $\mathbf{s}$-$\mathbf{t}$-flow $f$; the number of vertices with neighborhood given by $q$ informed by $v \in X$ is equal to $f_{v,q}$.
For the remaining unassigned vertices, we assign them greedily (and iteratively) to the earliest possible round in which an informed vertex is idle.
We remark that if any edge $(\mathbf{s}, h_v)$ is not saturated ($f_{\mathbf{s}, h_v} \neq c'_v$), then there is no protocol $P'$ that satisfies the properties given in the lemma.

The following claims show that the procedure is correct.

We say that a vertex $v$ is idle at time $i$ if
$v$ is connected to $s$ in $T'$,
$i > \min_{uv \in E(T')} \tau'(uv)$, and
$i \not\in \setdef{\tau'(uv)}{uv \in E(T')}$,
i.e., it is not currently assigned to transmit the message to another vertex at round $i$.
The procedure can be seen formally as follows.

\begin{enumerate}
\item The first stage of the procedure corresponds to adding an edge $a_vv$ and setting $\tau'(a_vv)=\tau_v$, for each $v \in X$. %
\item In the second stage, we add edges to $T'$ to make sure that each vertex $a_v$ is informed, that is, connected to $s$ in $T'$. %
For this purpose, we repeatedly take an $a_v$ in $A \cap R$ with minimum $\tau_v$, and find a vertex to inform it: %
if $b_v \neq \phi$, we simply find the earliest round $i<\tau_v$
where $b_v$ is idle and take $u=b_v$; %
if $b_v = \phi$, we find the earliest round $i<\tau_v$ with an idle $u \in V \setminus X$ that is informed by then;
In both cases, we add an edge $ua_v$ to $T'$ and set $\tau'(ua_v) = i$. %
\end{enumerate}

\begin{claim}
\label{claim:dtc:mechanism:1vtx}
In $P'$, each vertex transmits the message to at most one other vertex per round.
\end{claim}

\begin{claimproof}
In the first stage of the algorithm, we simply copy from $P$ the edges $a_vv$ and time-stamp $\tau'(a_vv)$ according to the information of \cref{lem:dtc:mechanism}.
Notice that if $P$ uses different informers for two vertices in $X$, then so does $P'$.
Thus, by the feasibility of $P$, each vertex can only inform one other per round.

For the second stage, by construction, we only pick a vertex $u$ to inform another if it is idle on that round, and thus the claim holds after this stage.

Finally, for the last stage, the max-flow assignment has a limit of at most $c'_v$ vertices for each $v$ to inform, which corresponds to the information of \cref{lem:dtc:mechanism}. %
Thus, each $v \in X$ can inform one vertex per round in the $c'_v$ rounds.
The remaining clique vertices are assigned greedily to idle vertices, thus ensuring that the property is satisfied.
\end{claimproof}

\begin{claim}
\label{claim:dtc:mechanism:complete}
The procedure completes successfully.
\end{claim}

\begin{claimproof}
We will start by reorganizing some parts of $P$ (for the purpose of the analysis) so that they match $P'$. %
First, we remark that we can exchange the role of two clique vertices completely, including when they were informed, as well as which vertices they inform and when, as long as they have the same neighborhood. %
Thus, we can assume that the vertices performing the role of informers, i.e.\ in the set $A$, are the same in $P$ and $P'$.

Next, we will modify $P$ such that the vertices in $A$ are informed at the same time as in $P'$.
For that purpose, we will execute partial switches between two vertices as follows:
let $u$ and $u'$ be two vertices and $i\geq i'$ be the times they are informed (in $P$), respectively,
such that $u'$ does not inform any vertices in $X$ before time $i$;
then, we can exchange $u$ and $u'$ so that $u$ is informed at time $i'$ (by the node that informed $u'$), $u'$ is informed at time $i$ (by the node that informed $u$), and $u$ takes over transmitting the messages of $u'$ before time $i$.
Since before round $i$, $u'$ only informed vertices in $V \setminus X$, this exchange is possible.

Now, consider the vertices of $A$ in the order of processing in the second stage.
We remark that this order is non-decreasing on the values $\iota_u = \min \setdef{\tau_v}{v \in X, a_v = u}$, the first time that a vertex $u$ informs a vertex in $X$.
For each $u \in A$ in this order, let $\tau$ be the time it is informed in $P$, $\tau'$ be its information time in $P'$, and $u'$ be the vertex in $P$ that occupies the same position as $u$ in $P'$, that is, the vertex in $P$ informed at time $\tau'$ by the informer of $u$ in $P'$.
If $u$ and $u'$ are the same vertex, there is nothing to do.
Otherwise, $u'$ cannot come before $u$ in the order, as then it would have already been processed and placed in the same position as in $P'$.
Thus, since $u'$ is either not in $A$ or after $u$ in the order, it cannot inform any vertex in $X$ before time $\tau \leq \iota_u \leq \iota_{u'}$.
Therefore, we can apply a partial switch to $u$ and $u'$ so that $u$ is informed at time $\tau'$ by the same informer as in $P'$.

This shows that the process in the second stage must complete successfully:
by induction, assuming that every vertex $u' \in A$ with $\iota_{u'} < \tau-1$ is informed by the same vertex and in the same round as $P$, then the vertex $a_v$ with minimum $\tau_v \geq \tau$ can be informed at the latest when it is informed in $P$, which by an exchange in $P$ means that it can be informed by the same vertex and in the same round as $P$.
Thus the process completes successfully.
\end{claimproof}

\begin{claim}
\label{claim:dtc:mechanism:causal}
In $P'$, each vertex only informs others after it is informed.
\end{claim}

\begin{claimproof}
We now show that, after the second stage, and since the procedure completes (by \cref{claim:dtc:mechanism:complete}),
$T'$ is a tree rooted at $s$ plus some isolated vertices, which implies that a vertex only informs others if it is informed itself. %

Let $T_s$ be the tree rooted at $s$ in $T'$. %
The proof is by induction on the minimum time $\tau_v$ of an $a_v \in R$, %
and the induction hypothesis states that, when a vertex $a_v \in A \cap R$ is considered,
every vertex $u \in A \cup X$ informed before $\tau_v$ is already connected to $s$ at that point of the algorithm.
Note that any vertex in $X$ is automatically in $T_s$ if $a_v$ is as well, and thus we focus the proof on the vertices in $A$.

When the first $a_v$ is picked, any vertex with label less than $\tau_v$ has to be in $T_s$, as it is either $s$ or a vertex informed by $s$; in this case, $a_v$ is either informed by $b_v$, which must be in $T_s$, or by a vertex in $V \setminus X$, which must be the root. %
In the general case, its informer $u$ is chosen as having a label $\tau'(u) < \tau_v$, and thus by induction is in the tree $T_s$; as a consequence, every vertex in the subtree of $T'$ rooted at $a_v$ is now part of $T_s$.
If there are multiple vertices $a_v \in A \cap R$ with the same value of $\tau_v$, they are handled in arbitrary order, and all added to $T_s$ by the argument above.
In any case, induction follows, and thus the vertices of $A \cup X$ are contained in $T_s$.

Notice that, by construction, we only pick a vertex $u$ to transmit the message to a vertex $a_v$ if it is in $T_s$, and after  it is informed.
Similarly, we pick a round $i$ to inform $a_v$ that is before $\tau_v$, so that $v$ is informed after $a_v$ is.
Overall, after the second stage, the claim holds for the vertices of $T_s$.

For the last stage, the claim also holds, since the assignments obtained from the max-flow problem assign $c'_v$ vertices for each $v$ to inform, allowing all of its assigned vertices to be informed by $v$ after time $\tau_v$. %
The remaining clique vertices are informed by vertices which are, by construction,  already informed, thus maintaining the claim.
\end{claimproof}

\begin{claim}
\label{claim:dtc:mechanism:info}
The protocol $P'$ respects all of the information of $P$ (in \cref{lem:dtc:mechanism}).
\end{claim}

\begin{claimproof}
The first stage of the procedure makes sure that vertices in $X$ are informed by the correct vertex (or at least, a vertex with the same neighborhood and informing the same subset of $X$), and at the right time.

It is also clear from the construction that each $a_v \in A$ is informed according to the information of \cref{lem:dtc:mechanism}, either by a clique vertex if $b_v = \phi$ or by $b_v$ otherwise.

Thus $P'$ respects all of the information of $P$ for vertices in $A \cup X$, and all that is left to show is that the remaining vertices are informed by time $t$.
Consider the protocol $P$ after the modifications in the proof of \cref{claim:dtc:mechanism:complete}.

By construction, the vertices assigned in the max-flow solution are all informed by time $t$, as they are informed by vertices in $X$, and each vertex $v \in X$ informs the same number of vertices as in $P$.

To show that the clique vertices in the final stage can be informed by time $t$, consider that before the assignments in the final stage, vertices in $A \cup X$ are occupied (i.e.\ inform some vertex) in the exact same rounds in $P$ and $P'$, due to the exchanges considered above.
Similarly, each vertex $v \in X$ informs others on the same rounds as in $P$, which are the earliest possible rounds after $\tau_v$ so that $c_v$ vertices are informed.

In $P'$, vertices in $\{s\} \cup A \cup X$ inform vertices of $V \setminus X$ in the same rounds as in $P$, and vertices in $V \setminus X$ can inform any other in the same set, so we conclude that if $P$ leads to every vertex being informed by time $t$, then so does $P'$.
\end{claimproof}
Combining the claims implies the correctness of the procedure, concluding the proof.
\end{proof}

We now apply \cref{lem:dtc:mechanism} to solve the problem.

The algorithm starts by computing a modulator $X$ of size at most $k$.
We remark that a modulator of $G$ is a vertex cover in $\bar G$, and thus the minimum modulator can be computed by employing any FPT algorithm for \textsc{Vertex Cover} (see e.g.~\cite{tcs/ChenKX10,stacs/HarrisN24}) on $\bar G$.

Then, it simply enumerates the information required by \cref{lem:dtc:mechanism}; for each possibility it uses the procedure in the lemma to obtain a protocol.
If there is a feasible solution to the problem, \cref{lem:dtc:mechanism} guarantees that the algorithm obtains a valid protocol when it uses the correct information, and thus the algorithm is correct.

As to the running time, it is dominated by the enumeration of all the options for \cref{lem:dtc:mechanism}:
	there are $\lceil\log_2 n\rceil + k$ possibilities for $t$ by \cref{lem:dtc:bound}, and
	the same number of possibilities for each $\tau_v$, and each $c_v$;
	for each $a_v$, there are $2k + 2^k$ options, corresponding to the vertices in $X$, their informers (if vertices share an informer), or a new vertex from one of each $2^k$ neighborhood classes;
	for each $b_v$, there are $k+1$ choices.
Thus, the total number of possibilities is
\[
	(\lceil\log_2 n\rceil + k)^{2k+1} \cdot (2k + 2^k)^k\cdot (k+1)^k = (\lceil\log_2 n\rceil)^{2k+1} \cdot k^{\bO(k)} \cdot 2^{k^2}.
\]

As the procedure in \cref{lem:dtc:mechanism} runs in polynomial time, the running time of the algorithm is FPT and bounded by $(2^k \log n)^{\bO(k)} n^{\bO(1)}$.
We can achieve the running time bound of $2^{\bO(k^2)}$ by case distinction:
if $\log n \leq 2^k$, then the bound follows by substitution;
otherwise, we substitute $k < \log \log n$ and get a running time of $2^{\bO((\log \log n)^2)} = 2^{\bO(\log n)} = n^{\bO(1)}$.

\subsection{Algorithm for Parameter Treewidth plus Time}
\label{sec:treewidth+time}

In this section we show, via a dynamic programming algorithm, that {\TB} admits
a single-exponential FPT algorithm if we parameterize by both treewidth and the
target number of rounds $t$.

\subparagraph{Treewidth and nice tree decompositions.}
A \emph{tree decomposition} of $G$ is a pair
$\mathcal{T}=(T,\{B_x\}_{x\in V(T)})$, where $T$ is a tree and each $B_x\subseteq V(G)$ is a \emph{bag},
such that the bags containing any vertex induce a nonempty connected subtree of $T$
and every edge has both endpoints in some bag.
Its width is $\max_{x\in V(T)}|B_x|-1$, and the \emph{treewidth} $\tw(G)$ is the minimum such width.

The decomposition is \emph{nice} if $T$ is rooted at a node $r$, the root and leaf bags are empty,
and every non-leaf node $x$ is one of three types:
an \emph{Introduce node} has one child $x'$ with $B_x=B_{x'}\cup\{v\}$ for some $v\notin B_{x'}$;
a \emph{Forget node} has one child $x'$ with $B_x=B_{x'}\setminus\{v\}$ for some $v\in B_{x'}$;
and a \emph{Join node} has two children $x_1,x_2$ with $B_x=B_{x_1}=B_{x_2}$.
Any width-$\tw$ tree decomposition can be transformed into a same-width nice tree decomposition
with $\bO(\tw\cdot n)$ nodes in time
$\bO(\tw^2\max\{|V(T)|,n\})$~\cite[Lemma~7.4]{books/CyganFKLMPPS15}.

\subparagraph{Overview.}
We next give an informal high-level overview of the main ideas.
The strategy of our algorithm will be to compute the proper edge-coloring which
assigns a time-stamp to each edge used by an optimal protocol. Intuitively, it
is natural that this leads to a complexity of the form $2^{t\cdot \tw}$, because
for each vertex $v$ of a bag we need to keep track of the set of colors which
have already been used on an edge incident on $v$. We therefore store a subset
of $[t]$ for each vertex of a bag. To facilitate the algorithm we will also
keep track of some further information, notably including the time at which a
vertex is informed by the protocol (but this only adds a further $t^{\tw}$
factor to the complexity). This extra information will allow us to infer the
direction in which each edge is supposed to be used (that is, which endpoint
transmits the message and which receives it) and verify that the protocol is
consistent (that is, each vertex starts transmitting the message after
receiving it).

\begin{theoremrep}\label{thm:fpt:twt}
    There is an algorithm which takes as input an $n$-vertex graph $G=(V,E)$, a
    vertex $s\in V$, an integer $t$, and a nice tree decomposition of $G$ of width $\tw$
    and determines whether $b(G,s)\le t$ in time $2^{\bO(t\cdot\tw)}n^{\bO(1)}$.
\end{theoremrep}

\begin{proof}

We reformulate the problem in a slightly more convenient way as follows: we
want to decide if there exists a set of edges $E' \subseteq E$, an orientation
$\vec{E}'$ of $E'$, and a coloring function $c \colon V \cup E' \to [0,t]$ with the
following properties:

\begin{enumerate}

\item All vertices other than $s$ have in-degree exactly $1$ in $\vec{E}'$,
while $s$ has in-degree $0$.

\item $c$ is a proper coloring of $E'$, that is, any two edges of $E'$ that
share an endpoint receive distinct colors.

\item It holds that $c(s) = 0$.

\item For all $v\in V\setminus\{s\}$, let $e = uv\in \vec{E}'$ be the unique arc
going into $v$. Then $c(v)=c(e)$.

\item For all $v \in V$, all arcs $e=vu\in\vec{E}'$ coming out of $v$ have $c(e)>c(v)$.

\end{enumerate}

\begin{claim}\label{claim:twt}
    If for an oriented subset of edges $\vec{E}'$ there exists a
    coloring $c$ satisfying the above, then $\vec{E}'$ forms an out-tree from $s$,
    that is, a digraph that contains a directed path from $s$ to all other vertices
    and whose underlying graph is a tree.
\end{claim}

\begin{claimproof} First, observe that $\vec{E}'$ cannot contain a directed
cycle. Indeed, for the sake of contradiction, suppose it does, and let $e=uv$
be an arc of this cycle of minimum color. Then, by the fourth property and the
selection of $e$ we have $c(u)\ge c(e)$, which violates the fifth property.

Hence, $\vec{E}'$ is a DAG where all vertices except $s$ have in-degree exactly
$1$. It must therefore have a forest as an underlying graph. To see this, we
observe that the digraph must contain a sink (as it is a DAG) and such a vertex
has total degree $1$. Repeatedly removing such vertices proves that the
underlying graph has no cycle.

We conclude that $\vec{E}'$ must contain a directed path from $s$ to all
vertices. For the sake of contradiction, suppose that $S \subseteq V$ is the set of
vertices not reachable from $s$. Notice that $|S| \neq 1$, as otherwise the single
vertex belonging to $S$ has an incoming edge from a vertex that is reachable from $s$,
a contradiction. Thus, assume that $|S| \ge 2$. Then, $\vec{E}'$ includes at least $|S|$ arcs
with both endpoints in $S$ because vertices of $S$ have in-degree $1$ and no
arc coming in from the part of the graph reachable from $s$. But, if there are
at least $|S|$ arcs with both endpoints in $S$, then the underlying graph
cannot be a forest, contradicting the previous paragraph.  \end{claimproof}

Given \cref{claim:twt}, it is now not hard to see that we can reformulate the
problem of deciding whether there exists a broadcast protocol from $s$ in $t$
rounds to the problem of finding an appropriate oriented subset $\vec{E}'$ of
the edges and a coloring $c$ satisfying the mentioned properties. To see the
equivalence, observe that if we have a protocol, we can place into $E'$ all
edges that are used at some point to transmit the message and orient them
towards the receiver, while we can define $c$ to encode the time when an edge
was used or when a vertex first received the message.
If the broadcast protocol
is non-redundant (that is, all vertices receive the message once) the properties are
satisfied. For the converse direction, given $\vec{E}'$ and $c$, we construct a
protocol by instructing each vertex $v$, for each arc $vu\in\vec{E}'$, to
transmit the message at time $c(vu)$ to vertex $u$. Because the coloring is
proper, no vertex transmits to two neighbors at the same time; because out-going
arcs have higher color all vertices transmit the message after they are
supposed to receive it; and because of \cref{claim:twt}, all vertices
eventually receive the message.

We now formulate a DP algorithm which decides if there exist $\vec{E}',c$
satisfying the properties we described, assuming we are given a nice tree
decomposition of $G$ of width~$\tw$. Let
$\mathcal{T}=(T,\{B_x\}_{x\in V(T)})$ be the given decomposition. We
denote by $B_x^{\downarrow}$ the set of all vertices of $G$ that appear in the
bag $B_x$ or in some bag that is a descendant of $x$.  We will define the
signature of a solution by giving for each $v\in B_x$ a triplet containing the
following pieces of information: (i)~an integer in $[0,t]$ representing $c(v)$,
(ii)~a boolean variable indicating whether the other endpoint of the arc
incoming to $v$ is in $B_x^{\downarrow}$, (iii)~a set $S\subseteq [0,t]$ containing
all the colors assigned to arcs of $\vec{E}'$ coming out of $v$ whose other
endpoint is in $B_x^{\downarrow}$. There are at most $2(t+1)2^{t+1}$ possible
triplets for a vertex $v$ and since a signature of a solution in a bag $B_x$ is
just a concatenation of the at most $|B_x|$ triplets, there are at most
$(2(t+1)2^{t+1})^{\tw+1}=2^{\bO(t \cdot \tw)}$ possible signatures.

Our DP algorithm now will construct a table for each node $x$ of the tree
decomposition. For each possible signature $\sigma$ we will store whether
$\sigma$ corresponds to the signature of a feasible partial solution in
$G[B_x^\downarrow]$. More precisely, we will say that $\sigma$ corresponds to a
feasible partial solution if and only if there exists an oriented subset
$\vec{E}'$ of edges of $G[B_x^{\downarrow}]$ and a coloring
$c \colon B_x^{\downarrow}\cup E'$ satisfying the following:

\begin{enumerate}

\item All vertices of $B_x^{\downarrow}$ have in-degree $1$ in $\vec{E}'$,
except $s$, which has in-degree $0$, and the vertices  $v\in B_x$ for which the
signature sets the boolean variable corresponding to $v$ to \texttt{false}, which also
have in-degree $0$.

\item $c$ is a proper coloring of $E'$.

\item For all vertices  $v\in B_x^\downarrow$ with in-degree $1$, if
$e\in\vec{E}'$ is the arc incoming to $v$, then $c(v)=c(e)$. Furthermore, all
arcs $e'=vu\in\vec{E}'$ coming out of $v$ have $c(e')>c(v)$.

\item For all vertices $v\in B_x$ and arcs $e=vu\in \vec{E}'$ coming out of
$v$ we have that $c(e)$ appears in the set of used colors stored in the third
component of the signature for $v$.

\item If $s \in B_x^{\downarrow}$ then $c(s)=0$.

\end{enumerate}

Observe now that for the root bag $r$ for which $B_r=\varnothing$ we only have
one signature to consider in our table (the empty signature). This will
correspond to a solution if and only if $\vec{E}',c$ exist satisfying our
desired properties, therefore if and only if a broadcast protocol of at most $t$ rounds exists, as
desired.

We now describe how the DP table can be populated from the bottom-up in a
straightforward way. For leaf nodes $x$ of the decomposition we have
$B_x^{\downarrow}=\varnothing$, so there is only one signature to consider (the
empty signature) and it always corresponds to a feasible partial solution.

For Introduce nodes $x$ with a child $x'$, we have $B_x=B_{x'}\cup\{v\}$ for
some $v\not\in B_{x'}$. For each signature $\sigma$ in the table for $B_{x'}$
that corresponds to a feasible partial solution, we will compute all the new
signatures that extend $\sigma$ and must be stored in the new table. To begin
with, our new table contains for each $\sigma$ from the table of $B_{x'}$, $t$
new signatures $\sigma_i, i\in[t]$, obtained by adding to $\sigma$ a triplet
for the new vertex $v$ that sets (i) $c(v)=i$, (ii) the Boolean variable for $v$
to \texttt{false}, (iii) the set of already used colors for $v$ to $\varnothing$. This
forms an initial table $T_0$.

We will now extend $T_0$ by considering one-by-one the edges $vu$ for each
$u\in B_x$ such that such an edge exists and exhausting all ways in which such
edges can be used.  Starting with a table $T_i$ for $B_x$, we consider a single
new edge, forming a larger table $T_{i+1}$. We continue in this way until all
edges have been considered and the final table is the table we store for $B_x$.

Suppose then that we take a signature $\sigma$ from the current table, and
$\sigma$ assigns to $u,v$ colors $c(u),c(v)$ respectively, the Boolean
variables indicating the in-degrees have values $b_u,b_v$, and the sets of used
colors are $S_u,S_v$. We first consider the possibility of using the edge $uv$
oriented towards $v$. This is only allowed if (i) $c(u)<c(v)$, (ii) $b_v$ is
\texttt{false}, (iii) $c(v)\not\in S_u$. If this is the case, we add to the table a new
signature that is the same as $\sigma$ except $b_v$ is now \texttt{true} and we have
placed $c(v)$ into $S_u$. Symmetrically, to consider the possibility of using
the edge $vu$ oriented towards $u$ we make the same check with the roles of
$u,v$ exchanged, and add the appropriate  modified signature after verification
of the conditions.

For Forget nodes $x$ with a child $x'$ such that $B_{x'}=B_x\cup\{v\}$ for
$v\not\in B_x$ we consider each signature $\sigma$ in the table of $B_{x'}$
such that $v$ has a Boolean variable assigned to \texttt{true}. For each such signature
we add to our table a new signature which is $\sigma$ restricted to $B_x$ (that
is, with the triplet for $v$ removed).

Finally, for Join nodes $x$ with two children $x_1,x_2$ such that
$B_x=B_{x_1}=B_{x_2}$ we consider each pair of signatures $\sigma_1,\sigma_2$,
one from each table and check if they can be combined to produce a signature in
the new table. We say that $\sigma_1,\sigma_2$ are \emph{compatible} if for all $v\in
B_x$ we have that (i) $\sigma_1,\sigma_2$ have the same color $c(v)$, (ii) at
most one of $\sigma_1,\sigma_2$ has the boolean variable for $v$ set to \texttt{true},
(iii) the sets of used colors for $v$ in $\sigma_1,\sigma_2$ are disjoint. If
$\sigma_1,\sigma_2$ are compatible, we form a new signature which for each $v$
(i) sets $c(v)$ to the same color as $\sigma_1,\sigma_2$, (ii) sets the boolean
variable for $v$ to \texttt{true} if at least one of $\sigma_1,\sigma_2$ does and to
\texttt{false} otherwise, (iii) sets the set of used colors for $v$ to be the union of
the sets given by $\sigma_1,\sigma_2$.

Correctness of the above algorithm can be verified by induction over the nice tree decomposition with the following invariant: the signatures stored for a given node $x$ are exactly the ones corresponding to a feasible partial solution. Furthermore,
the algorithm runs in time polynomial in the size of the decomposition and the
DP tables, so we obtain the promised running time.
\end{proof}

\section{Parameterized Approximation Results}
\label{sec:para-approx}

In this section, we present FPT approximation algorithms parameterized by clique cover number and cluster vertex deletion number. Both parameters are more general than distance to clique. The precise parameterized complexity for the two parameters remains unsettled.

\subsection{Clique Cover Number}\label{subsec:cliquecover}
For a graph $G$, a partition of its vertex set $V(G)$ into cliques is a \emph{clique cover}.
The minimum integer~$p$ such that $G$ admits a clique cover with $p$~cliques is the \emph{clique cover number}.
Computing an optimal clique cover of a graph is hard, as it is equivalent to properly coloring its complement with the minimum number of colors.
In this subsection, we show that the optimization version of {\TB} admits an FPT-AS
(i.e., an $f(p,\varepsilon)n^{\bO(1)}$-time $(1+\varepsilon)$-approximation algorithm for $\varepsilon>0$)
when parameterized by clique cover number $p$.

\begin{theoremrep}\label{thm:FPTAS:mmc}
    Given a connected $n$-vertex graph $G$, $s \in V(G)$, and a clique cover $Q_{1}, \dots, Q_{p}$ of $G$,
    one can find a broadcast protocol with at most $(1+\varepsilon) \cdot b(G,s)$ rounds in time $3^{2^{2 p/\varepsilon} p} \cdot n^{\bO(1)}$.
\end{theoremrep}

\begin{proof}
    If $|Q_{i}| \le 2^{2p/\varepsilon}$ for all $i \in [p]$, then $n \le 2^{2p/\varepsilon} \cdot p$.
    In this case, the $3^n n^{\bO(1)}$-time exact algorithm of Fomin et al.~\cite{tcs/FominFG24}
    runs in time $3^{2^{2p/\varepsilon} p} \cdot n^{\bO(1)}$.

    Otherwise, there exists a clique of size greater than $2^{2p/\varepsilon}$ in $G$, and thus $b(G,s) \ge \lceil\log \max_{i}|Q_{i}| \rceil > \log (2^{2p/\varepsilon})=  2p/\varepsilon$.
    For this case, we construct a two-phase broadcast protocol as follows.

    In the first phase, our goal is to inform at least one vertex in each clique $Q_{i}$.
    We claim that this can be done in at most $2p$ rounds.
    To see this, observe that given any set of currently informed vertices,
    if there exists a clique that contains no informed vertex,
    then there exists such a clique $Q_i$ that contains a vertex at distance at most $2$ from
    some informed vertex because $Q_{1}, \dots, Q_{p}$ are cliques.
    Therefore, we can make sure that a vertex of $Q_i$ receives the message in at most two rounds,
    increasing the number of cliques that contain an informed vertex.

    In the second phase, we broadcast the message independently in each clique $Q_{i}$ in parallel.
    This can be done greedily within at most $\lceil\log \max_{i}|Q_{i}| \rceil $ rounds.

    In total, the constructed broadcast protocol takes at most $2p + \lceil\log \max_{i}|Q_{i}| \rceil$ rounds.
    This proves the theorem as follows:
    \[
        \frac{2p + \lceil\log \max_{i}|Q_{i}| \rceil }{b(G,s)} \le \frac{2p + \lceil\log \max_{i}|Q_{i}| \rceil }{\lceil\log \max_{i}|Q_{i}| \rceil }
	\le 1 + \frac{2p}{2p/\varepsilon}
        = 1 + \varepsilon.
	\qedhere
    \]
\end{proof}

\subsection{Cluster Vertex Deletion Number}\label{subsec:cvd}
For a graph $G$, a set $S \subseteq V(G)$ is a \emph{cluster vertex deletion} if each connected component of $G-S$ is a complete graph. The \emph{cluster vertex deletion number} of $G$ is the size of a minimum cluster vertex deletion set of $G$.

\begin{theoremrep}\label{thm:FPTAS:cvd}
    Given a connected graph $G$ with cluster vertex deletion number $\cvd$ and $s \in V(G)$,
    one can find a broadcast protocol with at most $(2+\varepsilon) \cdot b(G,s)$ rounds in time $f(\cvd,\varepsilon) \cdot n^{\bO(1)}$.
\end{theoremrep}

\begin{proof}
    Let $S_0$ denote a cluster vertex deletion set of $G-s$ with size $|S_0| \le \cvd$;
    one can compute such a set in time $2^{\bO(\cvd)} n^{\bO(1)}$~\cite{mst/BoralCKP16,mst/HuffnerKMN10,mst/Tsur21}.
    Set $S = S_0 \cup \{s\}$.
    Furthermore, we make $S$ connected by adding at most $2\cvd$ vertices from $V(G) \setminus S$.
    To see this, notice that $G[S]$ has at most $\cvd + 1$ connected components, and since $G$ is connected,
    there exist (at most) two vertices belonging to a clique of $G-S$ such that adding them to $S$ connects two of its components.
    Exhaustively applying this process at most $\cvd$ times, we can make $S$ connected by adding at most $2\cvd$ vertices from $V(G) \setminus S$.
    Let $k = |S| \le 3\cvd + 1$.
    If every component of $G-S$ is of order at most $2^{k/\varepsilon}$, the vertex integrity of $G$ is at most $k+2^{k/\varepsilon}$.
    Since {\TB} is fixed-parameter tractable parameterized by vertex integrity (\cref{thm:fpt:vi}),
    it can be solved exactly in time $f(k+2^{k/\varepsilon})n^{\bO(1)}$,
    where $f$ depends only on $\cvd$ and $\varepsilon$.

    Otherwise, there exists a connected component of order more than $2^{k/\varepsilon}$ in $G-S$.
    Let $C^*$ be a connected component of $G-S$ with the maximum order.
    In this case, any broadcast protocol needs at least $\lceil\log |V(C^*)|\rceil$ rounds, that is, $b(G,s) \ge \lceil\log |V(C^*)|\rceil$.
    We construct an approximate broadcast protocol as follows in polynomial time.
    First, we send the message from $s$ to all the vertices in $S$.
    Since $S$ is connected and $s \in S$, this process terminates within $k$ rounds.

    Next, we send the message from $S$ to a vertex in each component of $G-S$.
    Let $\ell$ denote the minimum number of rounds required such that the vertices of $S$ inform at least one vertex in every component of $G-S$.
    We can compute $\ell$ by solving the \textsc{$b$-Matching} problem as follows.
    Consider the bipartite graph $B = (S\cup D, E(B))$ obtained from $G$ by contracting each connected component $C_i$ of $G-S$ into a single vertex $d_i \in D$
    and then forgetting the edges connecting vertices in $S$.
    Now, $\ell$ is equal to the minimum $\ell' \in [n]$ such that $B$ has a $b$-matching of size $|D|$,%
    \footnote{Given a graph $G$ and a function $b \colon V(G) \to [n]$,
    a \emph{$b$-matching} of $G$ is a set of edges $M\subseteq E$ such that every vertex $v$ is incident with at most $b(v)$ edges of $M$.}
    where $b(v) = 1$ for $v \in D$ and $b(v) = \ell'$ for $v \in S$.
    Consequently, in order to compute $\ell$ it suffices to solve a polynomial number of instances of the \textsc{$b$-Matching} problem,
    which is polynomial-time solvable~\cite{ipl/Anstee87,talg/Gabow18}.
    In our broadcast protocol, we send the message from $S$ to $G-S$ by following the obtained $b$-matching $M$ whenever it is possible.
    Specifically, we arbitrarily order the edges in $M$ incident with a vertex $v$ in $S$ and send the message from $v$ in this order.
    Clearly, this step needs at most $\ell$ rounds.
    Note that any broadcast protocol needs at least $\ell$ rounds, that is, $b(G,s) \ge \ell$.
    Once a vertex in a component of $G-S$ receives the message, for the rest of the rounds the vertex forwards the message towards the rest of the vertices
    of the same component if possible. This can be done in parallel in each component and takes at most $\lceil\log |V(C^*)|\rceil$ rounds.

    The obtained broadcast protocol takes at most $k + \ell + \lceil\log |V(C^*)|\rceil$ rounds, while $b(G,s) \ge \max\{\lceil\log |V(C^*)|\rceil, \, \ell\}$.
    This implies an approximation ratio of $2 + \varepsilon$ as
    \[
	\frac{k + \ell + \lceil\log |V(C^*)|\rceil}{\max\{\lceil\log |V(C^*)|\rceil, \, \ell\}}
        \le 2 + \frac{k}{\log 2^{k/\varepsilon}}
        = 2 + \varepsilon.
	\qedhere
    \]
\end{proof}

\section{Lower Bounds with Respect to Structural Parameters}\label{sec:structural_lb}

Here we present several results showcasing the relationship between the structure of the input graph
and the minimum number of rounds required by any broadcast protocol to complete.
On an intuitive level, graphs close to being stars or paths require a large
(that is, polynomial in the size of the graph)
number of rounds to broadcast the message,
and the aim of this section is to quantify this intuition and provide corresponding lower bounds with respect to
the tree-depth and the pathwidth of the input graph.
We start with three preliminary bounds that use, respectively, the number of components of $G-S$ (\cref{lem:structural_lb:cc}),
the diameter of $G$ (\cref{lem:structural_lb:diameter_edge_case}),
and the total diameter of the components of $G-S$ (\cref{lem:structural_lb:diameter}), where $S \subseteq V(G)$.
The tree-depth bound in \cref{thm:structural_lb:td} then uses the first lemma,
whereas the pathwidth bound in \cref{thm:structural_lb:pw_plus_length} applies all three across its three cases.
The latter theorem
improves over an analogous lower bound recently obtained by Aminian et al.~\cite{icalp/AminianKJS25}
and leads to \cref{corollary:structural_lb:pw,corollary:structural_lb:apx_ratio}.

\begin{lemma}\label{lem:structural_lb:cc}
    Given a connected graph $G$ and a non-empty set $S \subseteq V(G)$,
    for all $s \in V(G)$ it holds that
    $b(G,s) \ge |\cc(G-S)| / |S|$.
\end{lemma}

\begin{proof}
    Let a component of $G-S$ be \emph{active} when there exists a vertex belonging to it that has received the message.
    It holds that on every round the number of additional components of $G-S$ that become active is at most $|S|$.
    In that case, if $s \in S$, then at least $|\cc(G-S)| / |S|$ rounds are required in order to make all of the components of $G-S$ active.
    Notice that the lower bound holds in the case $s \notin S$ as well,
    since then, although there exists a component which is already active at the start,
    after the first round of the protocol no additional component of $G-S$ has become active.
\end{proof}

\begin{lemma}\label{lem:structural_lb:diameter_edge_case}
    Given a connected graph $G$,
    for all $s \in V(G)$ it holds that
    $b(G,s) \ge \diam(G)/2$.
\end{lemma}

\begin{proof}
    Consider a protocol of minimum broadcasting time $b(G,s)$
    as well as a shortest path $P$ in $G$ such that $\diam(G) = |V(P)|-1$
    with $u_1$ and $u_2$ denoting its endpoints.
    For every $v \in V(G)$ it holds that $d(v,s) \le b(G,s)$,
    that is, every vertex of $G$ is at distance at most $b(G,s)$ from $s$;
    if that were not the case $v$ cannot receive the message in time.
    Finally, since $P$ is a shortest path it holds that
    $\diam(G) = d(u_1,u_2) \le d(u_1,s) + d(s,u_2) \le 2 b(G,s)$
    and the statement follows.
\end{proof}


\begin{lemmarep}\label{lem:structural_lb:diameter}
    Given a connected graph $G$ and a non-empty set $S \subseteq V(G)$,
    for all $s \in V(G)$ it holds that
    \[
        b(G,s) \ge \sqrt{\frac{\sum_{C \in \cc(G-S)} \diam(C)}{2|S|}}.
    \]
\end{lemmarep}

\begin{proof}
    Consider a protocol of minimum broadcasting time $b(G,s)$.
    Let $\cc(G-S) = \setdef{C_i}{i \in [p]}$ for $p \in \mathbb{Z}^+$.
    For every $i \in [p]$, let $X_i = \setdef{x_{i,j}}{j \in [|X_i|]} \subseteq V(C_i)$ denote
    the vertices of component $C_i$ that receive the message directly from the vertices of $S$.
    Moreover, let for $q \in \{1,2\}$,
    $Y_{i,j}^q = \setdef{v \in V(C_i) \setminus X_i}{d_{C_i}(v,x_{i,j}) \le b(G,s) - q}$
    denote the vertices of $V(C_i) \setminus X_i$ that
    are at distance at most $b(G,s) - q$ from $x_{i,j}$ in $C_i$.
    In the following we consider two cases, depending on whether $s \in S$ or not.

    \proofsubparagraph{Case 1.}
    First assume that $s \in S$, in which case it holds that
    \begin{equation}\label{eq:structural_lb_1}
        \sum_{i \in [p]} |X_i| \le b(G,s) \cdot |S|,
    \end{equation}
    as the vertices of $S$ transmit the message to at most $|S|$ vertices per round.

    It holds that, for every vertex $v \in V(C_i) \setminus X_i$,
    there exists at least one vertex $x_{i,j_v} \in X_i$ such that $d_{C_i}(v,x_{i,j_v}) \le b(G,s) - 1$,
    that is, $v \in Y_{i,j_v}^1$;
    if that were not the case, $v$ cannot receive the message in time
    (the $-1$ is due to the round where $x_{i,j_v}$ receives the message from a vertex of $S$).
    Consequently, it follows that $\bigcup_{j \in [|X_i|]} Y_{i,j}^1 = V(C_i) \setminus X_i$.

    Consider a shortest path $P$ in $C_i$ such that $\diam(C_i) = |V(P)|-1$,
    while its intersection with $X_i$ is of maximum size.

    \begin{claim}\label{claim:structural_lb_diameter}
        For all $j \in [|X_i|]$, $|V(P) \cap Y_{i,j}^1| \le 2b(G,s) - 2$.
    \end{claim}

    \begin{claimproof}
        We first show that for all $j \in [|X_i|]$, $|V(P) \cap Y_{i,j}^1| \le 2b(G,s) - 1$.
        Assume otherwise, and let $j \in [|X_i|]$ such that $|V(P) \cap Y_{i,j}^1| > 2b(G,s) - 1$.
        Let $v_\ell, v_r \in V(P) \cap Y_{i,j}^1$ denote the leftmost and rightmost such vertices of $P$.
        It holds that $d_{C_i}(v_\ell, v_r) \le d_{C_i}(v_\ell, x_{i,j}) + d_{C_i}(x_{i,j}, v_r) \le 2b(G,s) - 2$.
        Furthermore, notice that since $P$ is a shortest path,
        it follows that $d_{C_i}(v_\ell, v_r) \ge |V(P) \cap Y_{i,j}^1| - 1 > 2b(G,s) - 2$, a contradiction.

        Now we consider the case where $|V(P) \cap Y_{i,j}^1| = 2b(G,s) - 1$.
        Notice that all those vertices appear consecutively in $P$;
        if not, $P$ is not a shortest path since
        $d_{C_i}(v_\ell, v_r) \le d_{C_i}(v_\ell, x_{i,j}) + d_{C_i}(x_{i,j}, v_r) \le 2b(G,s) - 2 < d_P(v_\ell, v_r)$.
        Then, one can substitute the $v_\ell - v_r$ subpath of $P$ by
        the paths $v_\ell - x_{i,j}$ and $x_{i,j} - v_r$, without increasing the length of the path,
        while increasing its intersection with $X_i$.
        Since $P$ is a path of maximum intersection with $X_i$, this leads to a contradiction.
    \end{claimproof}

    Recall that $\bigcup_{j \in [|X_i|]} Y_{i,j}^1 = V(C_i) \setminus X_i$.
    Therefore, due to \cref{claim:structural_lb_diameter} it holds that
    \begin{equation}\label{eq:structural_lb_2}
        \diam(C_i) <
        |X_i| + \sum_{j \in [|X_i|]} |V(P) \cap Y_{i,j}^1| \le
        |X_i| \cdot (2b(G,s) - 1) <
        2|X_i| \cdot b(G,s).
    \end{equation}
    The statement now follows from \cref{eq:structural_lb_1,eq:structural_lb_2}.

    \proofsubparagraph{Case 2.}
    In an analogous way we prove the statement in the case where $s \notin S$.
    Let $\pi \in [p]$ such that $s$ belongs to $C_{\pi} \in \cc(G-S)$.
    Notice that it holds that
    \begin{equation}\label{eq:structural_lb_3}
        \sum_{i \in [p]} |X_i| \le (b(G,s)-1) \cdot |S|,
    \end{equation}
    as the vertices of $S$ transmit the message to at most $|S|$ vertices per round,
    while at the first round no vertex of $S$ has received the message.

    For every $i \in [p] \setminus \{\pi\}$ it holds that,
    for every vertex $v \in V(C_i) \setminus X_i$,
    there exists at least one vertex $x_{i,j_v} \in X_i$ such that $d_{C_i}(v,x_{i,j_v}) \le b(G,s) - 2$,
    that is, $v \in Y_{i,j_v}^2$;
    if that were not the case, $v$ cannot receive the message in time
    (the $-2$ is due to the round where $x_{i,j_v}$ receives the message from a vertex of $S$,
    which in turn receives the message from one of its neighbors).
    Consequently, it follows that $\bigcup_{j \in [|X_i|]} Y_{i,j}^2 = V(C_i) \setminus X_i$,
    for all $i \in [p] \setminus \{\pi\}$.
    As for $C_{\pi}$, notice that if a vertex $v \in V(C_\pi)$ does not belong to $\bigcup_{j \in [|X_\pi|]} Y_{\pi, j}^2$,
    then it follows that $d_{C_{\pi}}(v,s) \le b(G,s)$; if that were not the case, then $v$ cannot receive the message by the end of round $b(G,s)$.
    Let $Y'_{\pi} = \setdef{v \in V(C_{\pi}) \setminus \{s\}}{d_{C_{\pi}}(v,s) \le b(G,s)}$ denote the set of such vertices,
    where $Y'_{\pi} \cup \{s\} \cup \bigcup_{j \in [|X_{\pi}|]} Y_{\pi,j}^2 = V(C_{\pi}) \setminus X_i$.

    In an analogous way as in \cref{claim:structural_lb_diameter} one can show that
    for all $i \in [p] \setminus \{\pi\}$ it holds that $\diam(C_i) < |X_i| \cdot (2b(G,s) - 3)$,
    while $\diam(C_{\pi}) < |X_{\pi}| \cdot (2b(G,s) - 3) + (1+2b(G,s))$.
    Summing over all $i \in [p]$ and making use of \cref{eq:structural_lb_3} gives
    \begin{align*}
        \sum_{i \in [p]} \diam(C_i)
        &< (2b(G,s)-3) \cdot \sum_{i \in [p]} |X_i| + (1+2b(G,s))\\
        &\le (2b(G,s)-3) \cdot (b(G,s)-1) \cdot |S| + (1+2b(G,s))\\
        &= 2|S| \cdot (b(G,s))^2 - 5|S| \cdot b(G,s) + 3|S| + 1 + 2b(G,s),
    \end{align*}
    thus it suffices to prove that $-5|S| \cdot b(G,s) + 3|S| + 1 + 2b(G,s) \le 0$,
    or equivalently that $2b(G,s)+1 \le |S| \cdot (5b(G,s)-3)$.
    Since $b(G,s) \ge 1$, it follows that this is equivalent to
    \[
        \frac{2b(G,s)+1}{5b(G,s)-3} \le |S|,
    \]
    which holds for any $b(G,s) \ge 2$, as $|S| \ge 1$.
    For the remaining case $b(G,s) = 1$, it follows that $G$ can only be a $K_2$,
    in which case the statement holds as well.
\end{proof}

We are now ready to proceed with the main result of this section concerning tree-depth.

\begin{theorem}\label{thm:structural_lb:td}
    Let $G$ be a connected graph with $n>1$ vertices of tree-depth $\td$.
    Then, for all $s \in V(G)$ it holds that $b(G,s) \ge n^{1 / \td} / \td$.
\end{theorem}

\begin{proof}
    Since $G$ is of tree-depth $\td$, there exists an \emph{elimination tree} $T$ of $G$
    of height at most $\td$. Let $r$ be the root of $T$.

    Notice that if $n^{1/\td} \le 2$, then the statement holds, unless $G$ is an isolated vertex.
    In the following, assume that $n^{1/\td} > 2$.
    We will prove that there exists a vertex with at least $n^{1/\td}$ children in $T$.
    Assume that this is not the case.
    Then, it holds that
    \[
        n = |V(T)| \le \sum_{i=0}^{\td-1} n^{(1/\td) \cdot i}
        = \frac{n - 1}{n^{1/\td}-1}
        < n-1,
    \]
    which is a contradiction.
    Consequently, there exists a vertex $u$ with at least $n^{1/\td}$ children in $T$.

    Let $S \subseteq V(G)$ contain the vertices present in the $u$-$r$ path of $T$, where $1 \le |S| \le \td$.
    Then, $G-S$ has at least $n^{1 / \td}$ connected components, 1 per child of $u$ in $T$,
    and the statement follows due to \cref{lem:structural_lb:cc}.
\end{proof}

Prior to presenting \cref{thm:structural_lb:pw_plus_length} we need the following definition.

\begin{definition}[{\cite[Proof sketch of Theorem~20]{icalp/AminianKJS25}}]
    A \emph{standard path decomposition} is a path decomposition in which the vertices in each bag
    are not a subset of the vertices in the preceding or succeeding bag along the path.
\end{definition}
Observe that in a standard path decomposition there are no two distinct bags with the vertices of the first one being
a subset of the vertices of the second.
We define the \emph{length} of a path decomposition to be the number of its bags.

\begin{theoremrep}\label{thm:structural_lb:pw_plus_length}
    Let $G$ be a connected graph with more than one vertex that admits a standard path decomposition of width $\pw$ and length $\ell$.
    Then, for all $s \in V(G)$ it holds that $b(G,s) \ge \sqrt{\ell^{1 / (\pw+1)} / 2\pw}$.
\end{theoremrep}

\begin{proof}
    Let $V(G) = \setdef{v_i}{i \in [n]}$.
    Consider a standard path decomposition $\mathcal{P}_1 = (P_1, \beta_1)$ of $G$ of width $\pw$ and length $\ell$,
    where for node $p \in V(P_1)$, $\beta_1(p) \subseteq V(G)$ denotes its bag,
    with $|\beta_1(p)| \le \pw+1$.
    Recall that no bag is a subset of another bag in $\mathcal{P}_1$,
    that is, $\beta_1(p) \not\subseteq \beta_1(p')$ for all distinct $p,p' \in V(P_1)$.
    We let the length of a path decomposition $\mathcal{P} = (P, \beta)$ be denoted by $L(\mathcal{P}) = |V(P)|$.

    For every vertex $v_i \in V(G)$,
    let $I^{\mathcal{P}}_i = \setdef{p \in V(P)}{v_i \in \beta(p)}$
    be the interval of bags in the path decomposition $\mathcal{P} = (P, \beta)$ containing $v_i$.
    We will say that vertex $v_i$ is \emph{present} in a path decomposition $\mathcal{P}$ if $I^{\mathcal{P}}_i \neq \varnothing$,
    while $\mathcal{I}^{\mathcal{P}} = \setdef{i \in [n]}{I^{\mathcal{P}}_i \neq \varnothing}$ is the set of the indices
    of all vertices of $G$ that are present in $\mathcal{P}$.

    Let %
    $G[\mathcal{P}]$ be the interval graph defined by the intervals
    $\setdef{I^{\mathcal{P}}_i}{i \in \mathcal{I}^{\mathcal{P}}}$ for a path decomposition $\mathcal{P}$.
    In that case, let $G' = G[\mathcal{P}_1] = (V,E')$
    be the interval graph defined by intervals $I^{\mathcal{P}_1}_1, \ldots, I^{\mathcal{P}_1}_n$,
    i.e., $G'$ is the graph obtained from $G$ by adding any edge $\{u,v\}$
    if there exists a node $p \in V(P_1)$ with $u,v \in \beta_1(p)$.
    Since $E' \supseteq E$, it follows that $b(G',s) \le b(G,s)$,
    and we will prove the statement for graph $G'$ instead.
    Notice that $\mathcal{P}_1$ is a standard path decomposition of graph $G'$.

    Having fixed our notation, we now provide a high-level overview of our proof.
    We will proceed inductively:
    starting from the interval graph $G[\mathcal{P}_j]$,
    as long as there exists some `long' interval $I^{\mathcal{P}_j}_i$
    (that is, $|I^{\mathcal{P}_j}_i| > L_j$ for some $L_j$) representing vertex $v_i$,
    we move on to the interval subgraph $G[\mathcal{P}_{j+1}]$,
    where $\mathcal{P}_{j+1}$ is defined by restricting $\mathcal{P}_j$ to
    the bags of $I^{\mathcal{P}_j}_i$ and removing vertex $v_i$.
    The described procedure will run for at most $\pw$ steps,%
    \footnote{For the sake of completeness we also consider the case where the procedure does not run for any steps,
    in which case \cref{lem:structural_lb:diameter_edge_case} can be applied.}
    resulting in a subset $S \subseteq V$ of size at most $\pw$,
    such that either $G'-S$ has many connected components,
    or the sum of the diameters of said connected components is large,
    thus allowing us to apply \cref{lem:structural_lb:cc,lem:structural_lb:diameter} respectively
    and obtain the stated bound.

    Let for $j \in [\pw]$, $L_j = \ell^{1 - \frac{j}{\pw+1}}$.
    We first note that if $L_{\pw} \le 2$, i.e., $\ell^{1/(\pw+1)} \le 2$,
    then $\sqrt{\ell^{1 / (\pw+1)} / 2\pw} \le 1 \le b(G,s)$ and the statement holds.
    Thus, in the following assume that $L_{\pw} > 2$.
    We set $S_1 = \varnothing$,
    and starting from $j=1$ we exhaustively apply the following rule.

    \proofsubparagraph{Rule $(\diamond)$.}
    If $j = \pw+1$, stop.
    Otherwise, if there exists $i \in \mathcal{I}^{\mathcal{P}_j}$ with $|I^{\mathcal{P}_j}_i| > L_j$,
    then set $S_{j+1} = S_j \cup \{v_i\}$,
    and $\mathcal{P}_{j+1} = (P_{j+1}, \beta_{j+1})$ to be the restriction
    of $\mathcal{P}_j$ in the bags contained in $I^{\mathcal{P}_j}_i$,
    albeit with removing the vertex $v_i$,
    i.e., $V(P_{j+1}) = V(P_j) \cap I^{\mathcal{P}_j}_i$ and
    $\beta_{j+1}(p) = \beta_j(p) \setminus \{v_i\}$, where $p \in V(P_{j+1})$.
    Notice that $\mathcal{I}^{\mathcal{P}_{j+1}} = \setdef{i' \in \mathcal{I}^{\mathcal{P}_{j}}}{I^{\mathcal{P}_j}_{i'} \cap I^{\mathcal{P}_j}_{i} \neq \varnothing} \setminus \{i\}$.
    Recurse on $j+1$.

    \begin{claim}\label{claim:structural_lb:pathwidth_helper}
        Let $j_{\max} \in [\pw+1]$ such that Rule $(\diamond)$ cannot be further applied.
        Then the following hold:
        \begin{enumerate}
            \item For $j \in [j_{\max}]$ and any two distinct nodes $p, p' \in V(P_{j})$,
            $\beta_j(p) \not\subseteq \beta_j(p')$, that is, $\mathcal{P}_j$ is a standard path decomposition.

            \item For all $p \in V(P_1)$, $S_{j_{\max}} \not\supseteq \beta_1(p)$.
        \end{enumerate}
    \end{claim}

    \begin{claimproof}
        We prove the first statement by induction.
        It is true for $j=1$ and assume that this is the case for $j-1$.
        Since $\mathcal{P}_j$ is obtained by considering a subset of bags of $\mathcal{P}_{j-1}$
        all of which contain a vertex which we remove, the statement holds.

        As for the second statement, it is clearly true if $j_{\max} = 1$ since $S_1 = \varnothing$.
        Assume that for $j_{\max} > 1$, there exists a bag $p'$ with $\beta_1(p') \subseteq S_{j_{\max}}$.
        For $S \subseteq V(G)$, let $B[S] = \setdef{p \in V(P_1)}{S \subseteq \beta_1(p)}$
        denote the bags of $\mathcal{P}_1$ that contain all elements of $S$.
        Notice that for all  $j \in [j_{\max}]$, it holds that $S_j \subseteq S_{j_{\max}}$, thus $B[S_{j}] \supseteq B[S_{j_{\max}}]$.
        Moreover, the vertex $v \in S_{j_{\max}} \setminus S_{j_{\max}-1}$ belongs to at least $L_{j_{\max}-1} \ge L_{\pw} \ge 2$ bags of $B[S_{j_{\max}-1}]$, thus $|B[S_{j_{\max}}]| \ge 2$, consequently there exists a bag
        $p \in B[S_{j_{\max}}] \setminus \{p' \}$ such that $\beta_1(p) \supseteq \beta_1(p')$,
        a contradiction.
    \end{claimproof}

    To finish our proof, we consider different cases depending on the value of $j_{\max}$.

    \begin{claim}\label{claim:structural_lb:pathwidth_case_cc}
        If $j_{\max} = \pw+1$, then $b(G',s) \ge {\ell^{1 / (\pw+1)} / \pw}$.
    \end{claim}

    \begin{claimproof}
        In the path decomposition $\mathcal{P}_1$ of $G'$
        there exist at least $L(\mathcal{P}_{j_{\max}})$ bags such that for every such bag $p$,
        $S_{j_{\max}} \subsetneq \beta(p)$ holds
        due to Rule~$(\diamond)$ and \cref{claim:structural_lb:pathwidth_helper}.
        Since each such bag has size at most $\pw+1$, while $|S_{j_{\max}}| = \pw$,
        it holds that each such bag is composed of the union of a distinct vertex of $G'$ and $S_{j_{\max}}$.
        In that case, $G' - S_{j_{\max}}$ has at least
        $L(\mathcal{P}_{j_{\max}}) > L_{\pw} = {\ell}^{1 / (\pw+1)}$
        connected components,
        and applying \cref{lem:structural_lb:cc} gives the stated bound.
    \end{claimproof}

    \begin{claim}\label{claim:structural_lb:pathwidth_case_diam_edge_case}
        If $j_{\max} = 1$, then $b(G',s) \ge \ell^{1 / (\pw+1)} / 2$.
    \end{claim}

    \begin{claimproof}
        Since $j_{\max} = 1$, it holds that for all $v_i \in V(G)$,
        $|I^{\mathcal{P}_{1}}_i| \le L_{1} = \ell^{\frac{\pw}{\pw+1}}$.
        In that case, it holds that $\diam(G') \ge \ell / L_{1} = \ell^{1 / (\pw+1)}$,
        as evidenced by the shortest path between vertices appearing only in the leftmost and the rightmost bags of $\mathcal{P}_1$
        (such vertices exist due to \cref{claim:structural_lb:pathwidth_helper}),
        as well as the fact that every vertex belongs to at most $L_{1}$ bags.
        The statement now follows by \cref{lem:structural_lb:diameter_edge_case}.
    \end{claimproof}

    \begin{claim}\label{claim:structural_lb:pathwidth_case_diam}
        If $j_{\max} \in [2, \pw]$, then $b(G',s) \ge \sqrt{{\ell}^{1 / (\pw+1)} / 2\pw}$.
    \end{claim}

    \begin{claimproof}
        Since $j_{\max} \in [2, \pw]$, it holds that for all $i \in \mathcal{I}^{\mathcal{P}_{j_{\max}}}$,
        $|I^{\mathcal{P}_{j_{\max}}}_i| \le L_{j_{\max}}$.
        We will prove that
        \[
            \sum_{C \in \cc(G'-S_{j_{\max}})} \diam(C) \ge \frac{L(\mathcal{P}_{j_{\max}})}{L_{j_{\max}}}.
        \]
        Notice that for all $j_{\max} \in [2,\pw]$,
        \[
            \frac{L(\mathcal{P}_{j_{\max}})}{L_{j_{\max}}} \ge
            \frac{L_{j_{\max}-1}}{L_{j_{\max}}} = {\ell}^{1/(\pw+1)},
        \]
        thus one can next apply \cref{lem:structural_lb:diameter} and obtain the stated bounds
        since $1 \le |S_{j_{\max}}| \le \pw$.

        Now consider the graph $G[\mathcal{P}_{j_{\max}}]$,
        and let $C_1, \ldots, C_q$, for $q \ge 1$, denote its connected components.
        These connected components induce a partition $(B_1, \ldots, B_q)$ on the bags of
        $\mathcal{P}_{j_{\max}}$,
        with a bag belonging to $B_i$ if its vertices lie in $C_i$.
        Assume that the components are numbered from left to right,
        with respect to the ordering of their bags in the path decomposition.
        Notice that for all $i \in [q]$, it holds that $\diam(C_i) \ge |B_i| / L_{j_{\max}}$,
        as evidenced by the shortest path between vertices appearing only in the leftmost and the rightmost bags of $B_i$
        (such vertices exist due to \cref{claim:structural_lb:pathwidth_helper}),
        as well as the fact that every vertex belongs to at most $L_{j_{\max}}$ bags.
        Moreover, as far as the connected components of graph $G'-S_{j_{\max}}$ are concerned,
        notice that $C_i \in \cc(G'-S_{j_{\max}})$ for all $i \in [2, q-1]$,
        while there exist components $C'_1, C'_2 \in \cc(G'-S_{j_{\max}})$ that extend $C_1$ and $C_q$ respectively,
        in which case the diameter of those components can only increase.

        Consequently, it holds that
        \[
            \sum_{C \in \cc(G'-S_{j_{\max}})} \diam(C) \ge \frac{\sum_{i \in [q]}|B_i|}{L_{j_{\max}}}=
            \frac{L(\mathcal{P}_{j_{\max}})}{L_{j_{\max}}}.
        \]
        This completes the proof.
    \end{claimproof}
    The statement now follows from \cref{claim:structural_lb:pathwidth_case_cc,claim:structural_lb:pathwidth_case_diam_edge_case,claim:structural_lb:pathwidth_case_diam}.
\end{proof}

Observe that one can easily transform a path decomposition into a standard path decomposition of the same width;
it suffices to discard any bag whose vertices are a subset of another bag's vertices.
Furthermore, notice that in any path decomposition every vertex appears at least once.
This, along with \cref{thm:structural_lb:pw_plus_length}, leads to \cref{corollary:structural_lb:pw},
which in turn, along with the approximation algorithm of Elkin and Kortsarz~\cite{jcss/ElkinK06},
leads to \cref{corollary:structural_lb:apx_ratio}.

\begin{corollary}\label{corollary:structural_lb:pw}
    Let $G$ be a connected graph with $n > 1$ vertices of pathwidth $\pw$.
    Then, for all $s \in V(G)$ it holds that
    \[
        b(G,s) \ge \sqrt{ \frac{(n/(\pw+1))^{1 / (\pw+1)}}{2\pw}}.
    \]
\end{corollary}

\begin{corollary}\label{corollary:structural_lb:apx_ratio}
    There is a polynomial-time algorithm for {\TB} achieving an approximation ratio of $\bO(\pw)$,
    where $\pw$ denotes the pathwidth of the input graph.
\end{corollary}




\section{Conclusion}\label{sec:conclusion}

In this paper we have given an extensive investigation of the complexity of the {\TB} problem with respect to
many of the most prominent structural graph parameters.
Let us also mention a few concrete open questions.
First, with regards to dense graph parameters, we have presented an exact FPT algorithm parameterized by the vertex deletion distance to clique,
but only approximation algorithms for the more general parameters cluster vertex deletion number and clique-cover number.
Can these be improved to exact FPT algorithms?

Second, with respect to treewidth, even though the problem is NP-hard for graphs of treewidth $2$,
we showed that an XP algorithm is possible if we seek a fast broadcast schedule, that is, $t=\bO(\log n)$,
because the problem admits an algorithm with parameter dependence $2^{\bO(t\cdot\tw)}$ (\cref{thm:fpt:twt}).
Can this be improved to an FPT algorithm parameterized by treewidth, under this restriction on $t$?
Note that treewidth is the only parameter for which this question is interesting,
as for graphs of small pathwidth we have established that the broadcast time must be significantly higher than $\bO(\log n)$ (\cref{corollary:structural_lb:pw}).

Finally, the computational complexity of {\TB} on graph classes would be another interesting research direction.
In particular, the complexity on split graphs remains open~\cite{tcs/JansenM95}.



\clearpage

\bibliography{bibliography}

\end{document}